%% file: main_isit.tex
\documentclass[11pt]{IEEEtran}

\usepackage{etoolbox}
\newtoggle{conference}
\togglefalse{conference} 

\usepackage[utf8]{inputenc}
\usepackage[T1]{fontenc}
\usepackage{url}
\usepackage{ifthen}
\usepackage{cite}
\usepackage[cmex10]{amsmath} 

\usepackage{hyperref}       
\usepackage{booktabs}       
\usepackage{amsfonts}       
\usepackage{nicefrac}       
\usepackage{microtype}      

\usepackage{amssymb, bm, epsfig, psfrag,amsthm}
\usepackage{algorithm,algorithmic}
\usepackage{bbm}
\usepackage[usenames]{color}
\usepackage{tikz}
\usepackage{enumerate}
\usetikzlibrary{shapes,arrows}
\usetikzlibrary{patterns}
\usepackage{url,hyperref}
\usepackage{xspace}
\usepackage{enumerate}
\usepackage{mathtools}
\usepackage{epstopdf}
\usepackage{eufrak}
\usepackage{sidecap}
\usepackage{graphicx}
\usepackage{subcaption}

\usepackage[normalem]{ulem} 

\newcommand{\citep}[1]{\cite{#1}}

\renewcommand{\hat}{\widehat}

\def\beq{\begin{equation}}
\def\eeq{\end{equation}}
\def\beqa{\begin{eqnarray}}
\def\eeqa{\end{eqnarray}}
\def\beqan{\begin{eqnarray*}}
\def\eeqan{\end{eqnarray*}}

\def\R{{\mathbb{R}}}

\def\argmin{\mathop{\mathrm{arg\,min}}}

\def\x{\times}

\newtheorem{theorem}{Theorem}

\newtheorem{assumption}{Assumption}

\def\zhat{\hat{z}}

\def\arr{\rightarrow}
\def\Exp{\mathbb{E}}

\def\Cov{\mathrm{Cov}}

\def\rank{\mathrm{rank}}
\def\alphabar{\overline{\alpha}}

\def\gammabar{\overline{\gamma}}

\def\tm1{t\! - \! 1}
\def\tp1{t\! + \! 1}
\def\km1{k\! - \! 1}
\def\kp1{k\! + \! 1}
\def\lp1{\ell\! + \! 1}
\def\lm1{\ell\! - \! 1}
\def\Lm1{L\! - \! 1}
\def\ip1{i\! + \! 1}
\def\im1{i\! - \! 1}

\newcommand{\zero}{\mathbf{0}}

\newcommand{\bbf}{\mathbf{b}}

\newcommand{\fbf}{\mathbf{f}}
\newcommand{\gbf}{\mathbf{g}}
\newcommand{\hbf}{\mathbf{h}}
\newcommand{\pbf}{\mathbf{p}}
\newcommand{\pbfhat}{\hat{\mathbf{p}}}
\newcommand{\qbf}{\mathbf{q}}
\newcommand{\qbfhat}{\hat{\mathbf{q}}}
\newcommand{\rbf}{\mathbf{r}}

\newcommand{\sbf}{\mathbf{s}}

\newcommand{\wbf}{\mathbf{w}}

\newcommand{\xbf}{\mathbf{x}}

\newcommand{\ybf}{\mathbf{y}}

\newcommand{\zbf}{\mathbf{z}}

\newcommand{\zbfhat}{\hat{\mathbf{z}}}
\newcommand{\Abf}{\mathbf{A}}

\newcommand{\Gbf}{\mathbf{G}}

\newcommand{\Kbf}{\mathbf{K}}

\newcommand{\Ubf}{\mathbf{U}}

\newcommand{\Vbf}{\mathbf{V}}

\newcommand{\Wbf}{\mathbf{W}}

\def\Lambdabar{\overline{\Lambda}}
\def\Sigmabf{{\boldsymbol \Sigma}}

\def\xibf{{\boldsymbol \xi}}

\newcommand{\phibf}{{\bm{\phi}}}

\newcommand{\tran}{^{\text{\sf T}}}

\def\PLeq{\stackrel{PL(2)}{=}}
\def\Norm{{\mathcal N}}

\def\Diag{\mathrm{Diag}}
\def\alphabar{\overline{\alpha}}

\newcommand{\bkt}[1]{{\langle #1 \rangle}}

\input{macros_parthe}

\title{Asymptotics of MAP Inference in Deep Networks}

\iftoggle{conference}{
    \author{
        \IEEEauthorblockN
        {Parthe Pandit\IEEEauthorrefmark{1},
        Mojtaba Sahraee\IEEEauthorrefmark{1},
        Sundeep Rangan\IEEEauthorrefmark{2},
        Alyson~K.~Fletcher\IEEEauthorrefmark{1}}
       \IEEEauthorblockA{\IEEEauthorrefmark{1}
        UCLA, Dept. Statistics and ECE, Los Angeles, CA,
        \{parthepandit,msahraee,akfletcher\}@ucla.edu}
       \IEEEauthorblockA{\IEEEauthorrefmark{2}
        NYU, Dept. ECE, Brooklyn, NY,
        srangan@nyu.edu}
    }
}
{
    \author{
        Parthe Pandit, Mojtaba Sahraee, Alyson~K.~Fletcher, Sundeep Rangan%
        \thanks{P. Pandit, M. Sahraee and A.~K.~Fletcher
        (email: \{parthepandit,msahraee,akfletcher\}@ucla.edu) are with
        the Department of Statistics and Electrical Engineering,
        the University of California, Los Angeles, CA, 90095\@.
        Their work was supported in part by the National Science Foundation under
        Grants 1254204 and 1738286, and the Office of Naval Research under Grant
        N00014-15-1-2677.
        S. Rangan (email: srangan@nyu.edu) is with
        the Department of Electrical and Computer Engineering,
        New York University, Brooklyn, NY, 11201\@.
        His work was supported in part by the National Science Foundation
        under Grants 1116589, 1302336, and 1547332,
        as well as
        the industrial affiliates of NYU WIRELESS.
    }
}
}

\begin{document}

\maketitle

\input{abstract}

\section{Introduction}
We consider inference in an $L$ layer stochastic neural network of the form,
\begin{subequations}  \label{eq:nntrue}
\begin{align}
    \zbf^0_\ell &= \Wbf_{\ell}\zbf^0_{\lm1} + \bbf_{\ell} + \xibf_\ell,
    \quad \ell=1,3,\ldots,\Lm1
        \label{eq:nnlintrue} \\
    \zbf^0_{\ell} &=  \phibf_\ell(\zbf^0_{\lm1},\xibf_{\ell}), \quad \ell = 2,4,\ldots,L.
        \label{eq:nnnonlintrue}
\end{align}
\end{subequations}
where $\zbf^0_0$ is the initial input, $\zbf^0_\ell$, $\ell=1,\ldots,L-1$ are the intermediate
hidden unit outputs and $\ybf=\zbf^0_L$ is the output.  The number of layers $L$ is even.
The equations \eqref{eq:nnlintrue} correspond to linear (fully-connected) layers
with weights and biases $\Wbf_\ell$ and $\bbf_\ell$, while \eqref{eq:nnnonlintrue} correspond to
elementwise activation functions such as sigmoid or ReLU.  The signals $\xibf_\ell$ represent noise
terms.  A block diagram for the network is shown in the top panel of Fig.~\ref{fig:nn_ml_vamp}.
The \emph{inference problem} is to estimate the initial and hidden states $\zbf^0_\ell$,
$\ell=0,\ldots,\Lm1$ from the final output $\ybf$.   We assume that network parameters
(the weights, biases and activation functions) are all known (i.e.\ already trained).
Hence, this is \emph{not} the learning problem.  The superscript 0 in $\zbf^0_\ell$ indicates that these
are the ``true" values, to be distinguished from estimates that we will discuss later.

This inference problem arises commonly when deep networks are used as generative priors.
Deep neural networks have been extremely successful in providing
probabilistic generative models of complex data such as images,
audio and text.  The models
can be trained either via variational autoencoders \citep{rezende2014stochastic,kingma2013auto}
or generative adversarial networks \cite{radford2015unsupervised,salakhutdinov2015learning}.
In inverse problems, a deep network is used as
a generative prior for the data (such as an image)
and additional layers are added to model the measurements (such as blurring, occlusion or noise)
\cite{yeh2016semantic,bora2017compressed}.
Inference can then be used to reconstruct the original image from the measurements.

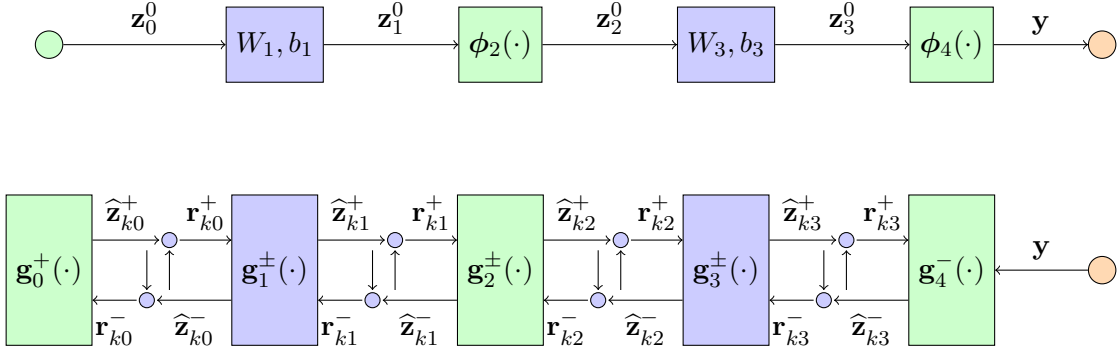
\begin{figure*}
\centering
\begin{tikzpicture}

    \pgfmathsetmacro{\sep}{3};
    \pgfmathsetmacro{\yoff}{0.4};
    \pgfmathsetmacro{\xoffa}{0.3};
    \pgfmathsetmacro{\xoffb}{0.6};

    \tikzstyle{var}=[draw,circle,fill=green!20,node distance=2.5cm];
    \tikzstyle{yvar}=[draw,circle,fill=orange!30,node distance=2.5cm];
    \tikzstyle{conn}=[draw,circle,fill=green!40,radius=0.02cm];
    \tikzstyle{linest}=[draw,fill=blue!20,minimum size=1cm,
        minimum height=2cm, node distance=\sep cm]
    \tikzstyle{nlest}=[draw,fill=green!20,minimum size=1cm,
        minimum height=2cm, node distance=\sep cm]
    \tikzstyle{linblock}=[draw,fill=blue!20, minimum size=1cm, node distance=\sep cm];
    \tikzstyle{nlblock}=[draw,fill=green!20, minimum size=1cm, node distance=\sep cm];

    \node [var] (z0) {};
    \node [linblock, right of=z0] (W1) {$W_1,b_1$};
    \node [nlblock, right of=W1] (phi2) {$\phibf_2(\cdot)$};
    \node [linblock, right of=phi2] (W3) {$W_3,b_3$};
    \node [nlblock, right of=W3] (phi4) {$\phibf_4(\cdot)$};
    \node [yvar,right of=phi4, node distance=2cm] (y) {};

    \path[->] (z0) edge  node [above] {$\zbf^0_0$} (W1);
    \path[->] (W1) edge  node [above] {$\zbf^0_1$} (phi2);
    \path[->] (phi2) edge  node [above] {$\zbf^0_2$} (W3);
    \path[->] (W3) edge  node [above] {$\zbf^0_3$} (phi4);
    \path[->] (phi4) edge  node [above] {$\ybf$} (y);

    \node [nlest,below of=z0] (h0) {$\gbf^+_0(\cdot)$};
    \node [linest,below of=W1] (h1) {$\gbf^{\pm}_1(\cdot)$};
    \node [nlest,below of=phi2] (h2) {$\gbf^{\pm}_2(\cdot)$};
    \node [linest,below of=W3] (h3) {$\gbf^{\pm}_3(\cdot)$};
    \node [nlest,below of=phi4] (h4) {$\gbf^-_4(\cdot)$};
    \node [yvar,right of=h4, node distance=2cm] (y2) {};
    \path [draw,->] (y2) edge node [above] {$\ybf$} (h4.east);

    \foreach \i/\j in {0/1,1/2,2/3,3/4} {
        \node [right of=h\i,xshift=\xoffa cm,yshift=\yoff cm] (conn0\i) {};
        \node [right of=h\i,xshift=\xoffb cm,yshift=\yoff cm] (conn1\i) {};
        \node [right of=h\i,xshift=\xoffa cm,yshift=-\yoff cm] (conn2\i) {};
        \node [right of=h\i,xshift=\xoffb cm,yshift=-\yoff cm] (conn3\i) {};
        \draw [fill=blue!20] (conn1\i) circle (0.1cm);
        \draw [fill=blue!20] (conn2\i) circle (0.1cm);

        \path [draw,->] (conn0\i) -- (conn2\i);
        \path [draw,->] (conn3\i) -- (conn1\i);

        \path[->] ([yshift=\yoff cm]h\i.east) edge node  [above]
            {$\widehat{\zbf}^+_{k\i}$} (conn1\i);
        \path[->] (conn1\i) edge node  [above]
            {$\rbf^+_{k\i}$} ([yshift=\yoff cm]h\j.west);
        \path[->] ([yshift=-\yoff cm]h\j.west) edge node [below]
            {$\widehat{\zbf}^-_{k\i}$} (conn2\i);
        \path[->] (conn2\i) edge node  [below]
            {$\rbf^-_{k\i}$} ([yshift=-\yoff cm]h\i.east);

    }


\end{tikzpicture}
\caption{Top panel:  Feedfoward neural network mapping an input $\zbf_0$ to output
$\ybf$ in the case of $L=4$ layers.
Bottom panel:  ML-VAMP inference algorithm for recovering estimates for the input
and hidden states from the output $\ybf$. \label{fig:nn_ml_vamp}}
\end{figure*}

Many deep network-based reconstruction methods perform maximum a priori (MAP) estimation
via minimization of the negative log likelihood \cite{yeh2016semantic,bora2017compressed}
or an equivalent regularized least-squares objective
\cite{chang2017one}.
MAP minimization is readily implementable
and has worked successfully in practice in problems such
as inpainting and compressed sensing.
MAP estimation also provides an alternative
to a separately learned reconstruction network
such as
\cite{mousavi2015deep,metzler2017learned,borgerding2017amp}.
However, due to the non-convex nature of the objective function,
MAP estimation has been difficult to analyze rigorously.  For example,
results such as \cite{hand2017global}
provide only general scaling laws while the guarantees in \cite{shah2018solving} require that
a non-convex projection operation can
be performed exactly.

To better understand MAP-based reconstruction,
this work considers inference in deep networks via
approximate message passing (AMP).  AMP  \cite{DonohoMM:09}
and its variants refer to a
powerful class of techniques for inverse problems that are both computationally efficient
and admit provable guarantees in certain high-dimensional limits.
Recent works
\cite{manoel2017multi,fletcher2018inference,gabrie2018entropy,reeves2017additivity}
have developed and analyzed variants of AMP for inference in multi-layer networks
such as \eqref{eq:nntrue}.
The methods generally consider minimum mean squared error (MMSE) inference
and estimation of the posterior density of the hidden units $\zbf_\ell$ from $\ybf$.
Similar to other AMP methods,
such MMSE-based multi-layer versions of AMP can be rigorously analyzed in cases with with large
random transforms.
This work specifically considers an extension of the multi-layer vector AMP
(ML-VAMP) method proposed in \cite{fletcher2018inference}.
ML-VAMP is derived from the recently-developed
VAMP method of \cite{rangan2017vamp,ma2017orthogonal,takeuchi2017rigorous} which is itself based on
expectation propagation \cite{minka2001expectation}
and expectation consistent approximate inference \cite{opper2005expectation,cakmak2014samp}.
Importantly, in the case of large random transforms, it is shown in \cite{fletcher2018inference}
that the reconstruction error  of ML-VAMP with MMSE estimation
can be exactly predicted, enabling much sharper results than other analysis techniques.
Moreover, under certain testable conditions ML-VAMP
can provably asymptotically achieve the Bayes optimal estimate, even for non-convex problems.

However, MAP estimation is often preferable to MMSE inference since MAP can be
formulated as an unconstrained optimization and implemented
easily via  standard deep learning optimizers
\cite{yeh2016semantic,bora2017compressed,chang2017one}.
This work thus considers a MAP version of ML-VAMP.
We show two key results.
First, it is shown that
the iterations in MAP ML-VAMP can be regarded as a variant
of an ADMM-type minimization \cite{boyd2011distributed} of the MAP objective.
This result is similar to earlier
connections between AMP and ADMM in~\cite{rangan2016fixed,rangan2017inference,manoel2018approximate}.
In particular, when MAP ML-VAMP converges, its fixed points are critical points of the
MAP objective.  Secondly, similar to the MMSE ML-VAMP considered in \cite{fletcher2018inference},
we can  rigorously analyze MAP ML-VAMP in a large system limit (LSL) with high-dimensional
random transforms $\Wbf_\ell$.  It is shown that, in the LSL,
the per iteration mean squared error of the
estimates can be exactly characterized by a state evolution (SE).  The SE tracks the correlation
between the estimates and true values at each layer and are only slightly more complex than
the SE updates for the MMSE case.  The SE enables an exact characterization of the
error of MAP estimation as a function of the network architecture,
parameters and noise levels.

\iftoggle{conference}{
For space considerations, all proofs are
contained in a full paper
\cite{pandit2019asymptotics-arxiv}.
In addition to the proofs, the full paper
includes further discussion with prior work,
as well as algorithm and simulation details.
}
{}

\begin{algorithm}[t]
\caption{ML-VAMP}
\begin{algorithmic}[1]  \label{algo:ml-vamp}
\REQUIRE{Forward estimation functions $\gbf_\ell^+(\cdot)$, $\ell=0,\ldots,\Lm1$ and
backward estimation functions $\gbf_\ell^-(\cdot)$, $\ell=1,\ldots,L$.}%
\STATE{Initialize $\rbf^-_{0\ell}=\zero$}
\FOR{$k=0,1,\dots,N_{\rm it}-1$}

    \STATE{// Forward Pass }
    \STATE{$\zbfhat^+_{k0} = \gbf_0^+(\rbf^-_{k0},\theta^+_{k0})$}
        \label{line:zp0}
    \STATE{$\alpha^+_{k0} = \bkt{\partial \gbf_0^+(\rbf^-_{k0},\theta^+_{k0})/
            \partial \rbf^-_{k\ell}}$}
            \label{line:alphap0}
    \STATE{$\rbf^+_{k0} = (\zbfhat^+_{k0} - \alpha^+_{k0}\rbf^-_{k0})/(1-\alpha^+_{k0})$}
            \label{line:rp0}
    \FOR{$\ell=0,\ldots,\Lm1$}
        \STATE{$\zbfhat^+_{k\ell} =
        \gbf_\ell^+(\rbf^+_{k,\lm1},\rbf^-_{k\ell},\theta_{k\ell}^+)$}
        \label{line:zp}
        \STATE{$\alpha^+_{k\ell} = \bkt{\partial
            \gbf_\ell^+(\rbf^+_{k,\lm1},\rbf^-_{k\ell},\theta_{k\ell}^+)/
            \partial \rbf^-_{k\ell}}$}
            \label{line:alphap}
        \STATE{$\rbf^+_{k\ell} = (\zbfhat^+_{k\ell} - \alpha^+_{k\ell}\rbf^-_{k\ell})/
            (1-\alpha^+_{k\ell})$}       \label{line:rp}
    \ENDFOR
    \STATE{}

    \STATE{// Reverse Pass }
    \STATE{$\zbfhat^-_{k,\Lm1} =
        \gbf_{L}^-(\rbf^+_{k,\Lm1},\theta^-_{k,L})$}   \label{line:znL}
    \STATE{$\alpha^-_{k,L} = \bkt{\partial
        \gbf_{L}^-(\rbf^+_{k,\Lm1},\theta^-_{k,L})/ \partial \rbf^+_{k,\Lm1}}$}
        \label{line:alphanL}
   \STATE{$\rbf^-_{\kp1,\Lm1} = (\zbfhat^-_{k,\Lm1}
            - \alpha^-_{k,\Lm1}\rbf^+_{k,\Lm1})/ (1-\alpha^-_{k,\Lm1}$)}
            \label{line:rnL}
    \FOR{$\ell=L-2,\ldots,0$}
        \STATE{$\zbfhat^-_{k\ell} =
            \gbf_{\lp1}^-(\rbf^+_{k\ell},\rbf^-_{\kp1,\lp1},\theta^-_{k,\lp1})$}
            \label{line:zn}
        \STATE{$\alpha^-_{k\ell} = \bkt{\partial
        \gbf_{\lp1}^-(\rbf^+_{k\ell},\rbf^-_{\kp1,\lp1},\theta^-_{k,\lp1}) /
            \partial \rbf^+_{k\ell}}$}
            \label{line:alphan}
        \STATE{$\rbf^-_{\kp1,\ell} = (\zbfhat^-_{k\ell}
            - \alpha^-_{k\ell}\rbf^+_{k\ell})/ (1-\alpha^-_{k\ell})$}
            \label{line:rn}
    \ENDFOR

\ENDFOR
\end{algorithmic}
\end{algorithm}

\section{ML-VAMP for MAP Inference}
We consider inference in a probabilistic setting where,
in \eqref{eq:nntrue}, $\zbf^0_0$ and $\xibf_\ell$ are modeled as random vectors
with some known densities.  Inference can be then performed by MAP estimation,
\beq \label{eq:Jmapopt}
    \zbfhat = \underset{\zbf}\argmin\  J(\zbf,\ybf),
\eeq
where $J(\zbf,\ybf)$ is the negative log posterior,
\beq \nonumber
    J(\zbf,\ybf) := -\ln p(\zbf_0) - \sum_{\ell=1}^{\Lm1} \ln p(\zbf_\ell|\zbf_{\lm1})
        -\ln p(\ybf|\zbf_{\Lm1}),
\eeq
where $p(\zbf_0)$ is the prior on the initial input $\zbf^0_0$ and
$\ln p(\zbf_\ell|\zbf_{\lm1})$ is defined implicitly from the probability distribution on the
noise terms $\xibf_\ell$ and the updates in \eqref{eq:nntrue}.

The ML-VAMP algorithm from \cite{fletcher2018inference}
for the inference problem is shown in Algorithm~\ref{algo:ml-vamp}.
For each hidden output $\zbf_\ell$, the algorithm produces two estimates $\zbfhat^+_{k\ell}$
and $\zbfhat^-_{k\ell}$ indexed by the iteration number $k$.  In each iteration, there is a
forward pass that produces the estimates $\zbfhat^+_{k\ell}$ and a reverse pass that produces
the estimates $\zbfhat^-_{k\ell}$.  The estimates are produced by a set of \emph{estimation
functions} $\gbf_\ell^{\pm}(\cdot)$ with parameters
$\theta^{\pm}_{k\ell}$.   The recursions are illustrated in the bottom panel of
Fig.~\ref{fig:nn_ml_vamp}.

For MAP inference, we propose the following estimation functions $\gbf_\ell^{\pm}(\cdot)$:
For $\ell=1,\ldots,L-2$, let
$\theta_\ell = (\gamma_{\lm1}^+,\gamma^-_\ell)$, and define the energy function,
\begin{align}
    \MoveEqLeft
    J_\ell(\zbf_{\lm1}^{ {-}{}},\zbf_\ell^{{+}{}};\rbf^+_{\lm1},\rbf^-_\ell,\theta_\ell)
    := -\ln p(\zbf_\ell^{ {+}{}}|\zbf_{\lm1}^{ {-}{}})
    \nonumber \\
    &+
    \frac{\gamma^+_{\lm1}}{2}\|\zbf_{\lm1}^{ {-}{}}-\rbf^+_{\lm1}\|^2
    + \frac{\gamma^-_{\ell}}{2}\|\zbf_{\ell}^{ {+}{}}-\rbf^-_{\ell}\|^2.\label{eq:defineEnergy_l}
\end{align}
In the MMSE inference problem considered in \cite{fletcher2018inference},
the estimation functions $\gbf_\ell^{\pm}$ are given by the expectation with respect to the joint density,
$p(\zbf_{\lm1}^{ {-}{}},\zbf_\ell^{ {+}{}}) \propto \exp[-J_\ell(\cdot)]$.
In this work, we consider the MAP
estimation functions given by the mode of this density:
\begin{align} \label{eq:gmap}
    \left((\gbf_{\ell}^-(\rbf^+_{\lm1},\rbf^{ {-}{+}}_\ell,\theta_\ell),
    \gbf_\ell^+(\rbf^+_{\lm1},\rbf^{ {-}{+}}_\ell,\theta_\ell)\right)
        := (\zbfhat^-_{\lm1},\zbfhat^+_{\ell})
\end{align}
where
\beq\label{eq:mapEnergy_l}
    (\zbfhat^-_{\lm1},\zbfhat^+_\ell) = \argmin_{\zbf_{\lm1}^-,\zbf_\ell^+}
        J_\ell(\zbf_{\lm1}^{ {-}{}},\zbf_\ell^{ {+}{}};\rbf^+_{\lm1},\rbf^-_\ell,\theta_\ell).
\eeq
Similar equations hold for $\ell=0$ and $\ell=\Lm1$ by removing the terms for $\ell=0$ and $L$.

In the MMSE inference in \cite{fletcher2018inference}, the parameters $\theta_{k\ell}^\pm$ are selected
as,
\beq \label{eq:thetagam}
    \theta_{k\ell}^+ = (\gamma_{k,\lm1}^+,\gamma^-_{k\ell}),
    \quad
    \theta_{k\ell}^- = (\gamma_{\kp1,\lm1}^+,\gamma^-_{k\ell}),
\eeq
where the precision levels $\gamma_{k\ell}^\pm$ are updated by the recursions,
\begin{align} \label{eq:gamupdate}
\begin{split}
    \gamma^+_{k\ell} &= \eta^+_{k\ell} - \gamma^-_{k\ell}, \quad
            \eta^+_{k\ell} = \gamma^-_{k\ell}/\alpha^+_{k\ell} \\
    \gamma^-_{\kp1,\ell} &= \eta^-_{k\ell} - \gamma^+_{k\ell}, \quad
            \eta^-_{k\ell} = \gamma^+_{k\ell}/\alpha^-_{k\ell}.
\end{split}
\end{align}
We can use the same updates for MAP ML-VAMP, although some
of our analysis will apply to arbitrary parameterizations.

\section{Fixed Points and Connections to ADMM} \label{sec:fixed}

Our first results relates MAP ML-VAMP to an ADMM-type minimization of the
MAP objective \eqref{eq:Jmapopt}.  To simplify the presentation, we consider MAP estimation functions
\eqref{eq:gmap} with \emph{fixed values} $\gamma^{\pm}_\ell>0$.  Also, we replace
the $\alpha^{\pm}_{k\ell}$ updates in Algorithm~\ref{algo:ml-vamp} with fixed values,
\beq \label{eq:alphafix}
    \alpha^+_\ell=\gamma^-_\ell/\eta_\ell, \ \
    \alpha^-_\ell=\gamma^+_\ell/\eta_\ell,\ \ {\rm and}\ \  \eta_\ell = \gamma^+_\ell+\gamma^-_\ell.
\eeq
Now, to apply ADMM \cite{boyd2011distributed} to the  MAP optimization~\eqref{eq:Jmapopt}, we use \emph{variable splitting} where
we replace each variable $\zbf_\ell$ with two copies $\zbf_\ell^+$ and $\zbf^-_\ell$.
Then, we define the objective function,
\begin{align}
   \MoveEqLeft F(\zbf^+,\zbf^-) := -\ln p(\zbf^+_0)  \nonumber \\
    &- \sum_{\ell=1}^{\Lm1} \ln p(\zbf^+_\ell|\zbf^-_{\lm1})
   -\ln p(\ybf|\zbf_{\Lm1}^{-}), \label{eq:Fsplit}
\end{align}
over the groups of variables $\zbf^{\pm} = \{ \zbf^\pm_\ell \}$.
The minimization in~\eqref{eq:Jmapopt} is then equivalent
to the constrained optimization,
\begin{align} \label{eq:Fmincon}
    \min_{\zbf^+,\zbf^-} \ F(\zbf^+,\zbf^-)
    \mbox{ s.t. }\
    \zbf^+_\ell=\zbf^-_\ell\ ~\forall\ \ell.
\end{align}

Corresponding to this constrained optimization, define the augmented Lagrangian,
\begin{align}
    \mc L(\zbf^+,\zbf^-,\sbf) = &F(\zbf^+,\zbf^-)
    + \sum_{\ell={0}}^{\Lm1} \eta_\ell\sbf\tran_\ell(\zbf^+_\ell-\zbf_\ell^-)\nonumber\\ &+\sum_{\ell=0}^{L-1}\frac{\eta_\ell}{2}\|\zbf_{\ell}^+-\zbf_{\ell}^-\|^2,
    \label{eq:Lagdef}
\end{align}
where $\sbf=\{\sbf_\ell\}$ are a set of dual parameters and $\gamma_\ell^{\pm}>0$ are weights
and $\eta_\ell = \gamma^+_\ell+\gamma^-_\ell$.   Now, for $\ell=1,\ldots,L-2$, define
\begin{align}
    \MoveEqLeft \mc L_\ell(\zbf_{\lm1}^-,\zbf_\ell^+;\zbf_{\lm1}^+,\zbf_\ell^-,\sbf_{\lm1},\sbf_\ell)
    := -\ln p(\zbf_\ell^+|\zbf_{\lm1}^-) \nonumber \\
    & +\eta_\ell\sbf\tran_\ell\zbf^+_\ell -\eta_{\lm1}\sbf\tran_{\lm1}\zbf_{\lm1}^- \nonumber \\
    &+ \frac{\gamma^+_{\lm1}}{2}\|\zbf^-_{\lm1} - \zbf^+_{\lm1}\|^2
    + \frac{\gamma^-_\ell}{2}\|\zbf^+_{\ell} - \zbf^-_{\ell}\|^2,  \label{eq:Laug}
\end{align}
which represents the terms in the Lagrangian $\mc L(\cdot)$ in \eqref{eq:Lagdef}
that contain $\zbf_{\lm1}^-$ and $\zbf_\ell^+$.
Similarly, define $\mc L_0(\cdot)$ and $\mc L_{\Lm1}(\cdot)$ using $p(\zbf_0^+)$ and $p({\bf y}| \zbf^+_{L-1})$.
\iftoggle{conference}{}{One can verify that
\beq\nonumber
\mc L(\zbf^+,\zbf^-,\sbf) = \sum_{\ell=0}^{L-1}\mc L_\ell(\zbf_{\lm1}^-,\zbf_\ell^+;\zbf_{\lm1}^+,\zbf_\ell^-,\sbf_{\lm1},\sbf_\ell).
\eeq
}

\begin{theorem} \label{thm:mapfix}
Consider the outputs of the {\rm ML-VAMP} (Algorithm~\ref{algo:ml-vamp}) with MAP estimation functions
\eqref{eq:gmap} for \emph{fixed} $\gamma_\ell^{\pm}>0$.
Suppose lines \ref{line:alphap} and \ref{line:alphan} are replaced with fixed values $\alpha^{\pm}_{k\ell}=\alpha^{\pm}_{\ell}\in(0,1)$ from \eqref{eq:alphafix}.
Let,
\begin{align}\label{eq:defines}
\sbf^{-}_{k\ell} := \alpha^{+}_{k\ell}(\zbfhat_{\km1,\ell}^--\rbf^-_{k\ell}), \quad
\sbf^{+}_{k\ell} := \alpha^{-}_{k\ell}(\rbf^+_{k\ell}-\zbfhat_{k\ell}^+).
\end{align}
Then, the forward pass iterations satisfy,
\begin{subequations}
\begin{align}
    \underline{\hspace{.3cm}}\,,\zbfhat^+_{k\ell} &= \argmin_{(\zbf_{\lm1}^-,\zbf^+_{\ell})}\
        \mc L_\ell(\zbf^-_{\lm1},\zbf^+_\ell;\zbfhat^+_{k,\lm1},\zbfhat^-_{\km1,\ell},\sbf_{k,\lm1}^+,\sbf_{k\ell}^-)  \label{eq:admmHp}\\
    \sbf_{k\ell}^+ &= \sbf_{k\ell}^- + \alpha^+_\ell(\zbfhat^+_{k\ell}-\zbfhat^-_{\km1,\ell}). \label{eq:admmsp}
\end{align}
\end{subequations}
whereas the backward pass iterations satisfy,
\begin{subequations}
\begin{align}
    \MoveEqLeft \zbfhat^-_{k,\lm1},\,\underline{\hspace{.3cm}}\, \nonumber \\
        &= \argmin_{(\zbf_{\lm1}^-,\zbf^+_{\ell})}\
        \mc L_\ell(\zbf^-_{\lm1},\zbf^+_\ell;\zbfhat^+_{k,\lm1},\zbfhat^-_{k\ell},
            \sbf_{k,\lm1}^+,\sbf_{\kp1,\ell}^-)  \label{eq:admmHn} \\
    \MoveEqLeft \sbf_{\kp1,\lm1}^- = \sbf_{k,\lm1}^+ + \alpha^-_{\lm1}(\zbfhat^+_{k,\lm1}
        -\zbfhat^-_{k,\lm1}).
    \label{eq:admmsn}
\end{align}
\end{subequations}
for $\ell=0,\ldots,L-1$. Further, any fixed point of Algorithm 1 corresponds to a critical point of the Lagrangian \eqref{eq:Lagdef}.
\end{theorem}
\iftoggle{conference}{}{
\begin{proof}  See Appendix~\ref{sec:mapfixpf}.
\end{proof}
}

As shown in the above result, the fixed $(\alpha_\ell^{\pm})$ version of ML-VAMP is an ADMM-type algorithm for solving the optimization problem \eqref{eq:Fmincon}.
\iftoggle{conference}{}{For $\alpha_\ell^+=\alpha^-_\ell,$ its convergence properties have been studied extensively under the name Peaceman-Rachford Splitting Method (PRSM) (see \cite[eqn. (3)]{he2016application} and \cite[eqn. (1.12)]{han2012convergence}, and the references therein).}
The full ML-VAMP algorithm adaptively updates $(\alpha_{k\ell}^{\pm})$ to  the take into account information regarding the curvature of the objective in \eqref{eq:gmap}. Note that in \eqref{eq:admmHp} and \eqref{eq:admmHn}, we compute the joint minima over $(\zbf^+_{\lm1},\zbf^+_\ell)$, but only use one of them at a time.


\section{Analysis in the Large System Limit} \label{sec:seevo}

As mentioned in the Introduction, the paper \cite{fletcher2018inference} provides
an analysis of ML-VAMP with MMSE estimation functions
in a certain large system limit (LSL).
We extend this analysis
to general estimators, including the MAP estimators~\eqref{eq:gmap}.
The LSL analysis has the same basic assumptions as \cite{fletcher2018inference}.
\iftoggle{conference}{Details are in the full paper and can be summarized as follows.}{
Details of the assumptions are given in Appendix~\ref{sec:lsl}.  The key assumptions are summarized
as follows.

}
We consider a sequence of problems indexed by $N$.
For each $N$, and $\ell=1,3,\ldots,\Lm1$, suppose that the weight matrix $\Wbf_\ell$ has the SVD
\beq \label{eq:WSVD}
    \Wbf_{\ell} = \Vbf_{\ell}\Sigmabf_\ell\Vbf_{\lm1}, \quad
    \Sigmabf_{\ell} =
    \left[ \begin{array}{cc}
    \Diag(\sbf_\ell) & \zero \\
    \zero & \zero  \end{array} \right] \in \R^{N_\ell \x N_{\lm1}},
\eeq
where
$\Vbf_\ell$ and $\Vbf_{\lm1}$ are orthogonal matrices,
the vector $\sbf_\ell = (s_{\ell 1},\ldots,s_{\ell R_\ell})$ contains
singular values,
and $\rank(\Wbf_\ell)\leq R_{\ell}$.
Also, let $\bar{\bbf}_\ell := \Vbf_\ell\tran\bbf_\ell$ and
\iftoggle{conference}{
$\bar{\xibf}_\ell := \Vbf_\ell\tran \xibf_\ell$.}{
$\bar{\xibf}_\ell := \Vbf_\ell\tran \xibf_\ell$ so that
\beq \label{eq:bxibar}
    \bbf_\ell = \Vbf_\ell\bar{\bbf}_\ell, \quad
    \xibf_\ell = \Vbf_\ell\bar{\xibf}_\ell.
\eeq
}
The number of layers $L$ is fixed and
the dimensions $N_\ell = N_\ell(N)$ and ranks $R_\ell = R_\ell(N)$ in each layer
are deterministic functions of $N$.
We assume that
$\lim_{N \arr \infty} N_\ell/N$ and $\lim_{N \arr \infty} R_\ell/N$ converge to non-zero constants,
so that the dimensions grow linearly with $N$.

For the estimation functions in the linear layers $\ell=1,3,\ldots,L-1$,
we assume that they are the MAP estimation functions \eqref{eq:gmap},
but the parameters $\gamma^+_{\lm1}$ and $\gamma^-_\ell$ can be chosen arbitrarily.
Since the conditional density $p(\zbf_\ell|\zbf_{\lm1})$ is given by the linear
update \eqref{eq:nnlintrue}, the MAP estimation function \eqref{eq:gmap}
is identical to the MMSE function and is given by a solution to a least squares problem.
For the nonlinear layers, $\ell=0,2,\ldots,L$, the estimation functions $\gbf_\ell(\cdot)$
can be arbitrary as long as they operate elementwise and are Lipschitz continuous.
For simplicity, we will assume that for all the estimation functions, the parameters
$\theta_{k\ell}$ are deterministic and fixed.  However, data dependent parameters
can also be considered as in \cite{fletcher2017inference}.

We follow the analysis methodology in \cite{BayatiM:11},
and assume that the signal realization $\zbf^0_\ell \in \R^{N_0}$ for $\ell=0$,
and the noise realizations $\xibf_\ell$ in the nonlinear stages $\ell=2,4,\ldots,L$,
\iftoggle{conference}{
all converge empirically to random variables $Z^0$ and $\Xi_\ell$.
For the linear stages $\ell=1,3,\ldots,\Lm1$,
let $\bar{\sbf}_\ell$ be the zero-padded singular value vector,
and we assume that $(\bar{\sbf}_\ell, \bar{\bbf}_\ell, \bar{\xibf}_\ell)$ also
converge empirically.}
{all converge empirically to random variables $Z^0$ and $\Xi_\ell$, i.e.,
\beq \label{eq:varinitnl}
    \lim_{N \arr \infty} \left\{ z^0_{0,n} \right\} \PLeq Z^0_0, \quad
    \lim_{N \arr \infty} \left\{ \xi_{\ell,n} \right\} \PLeq \Xi_\ell.
\eeq
Convergence $PL(2)$ is reviewed in Appendix~\ref{sec:empirical} -- see
\cite{BayatiM:11,fletcher2017inference} and elsewhere.
For the linear stages $\ell=1,3,\ldots,\Lm1$,
let $\bar{\sbf}_\ell$ be the zero-padded singular value vector,
\beq \label{eq:sbar}
    \bar{s}_{\ell,n}= \begin{cases}
        s_{\ell,n} & \mbox{if } n=1,\ldots,R_\ell, \\
        0          & \mbox{if } n=R_\ell+1,\ldots,N_\ell,
    \end{cases}
\eeq
so that $\bar{\sbf}_\ell \in \R^{N_\ell}$.
We assume that $\bar{\sbf}_\ell$,
the transformed bias $\bar{\bbf}_\ell=\Vbf_\ell\tran\bbf_\ell$, and the transformed noise
$\bar{\xibf}_\ell=\Vbf_\ell\tran\xibf_\ell$ all converge empirically as
\beq \label{eq:varinitlin}
    \lim_{N \arr \infty} \left\{ \bar{s}_{\ell,n},\bar{b}_{\ell,n},\bar{\xi}_{\ell,n} \right\}
        \PLeq (\bar{S}_\ell, \bar{B}_\ell,\bar{\Xi}_\ell),
\eeq
to independent random variables $\bar{S}_\ell$, $\bar{B}_\ell$, and $\bar{\Xi}_\ell$, with
$\bar{\Xi}_\ell \sim \Norm(0,\nu_\ell^{-1})$, where
$\nu_\ell$ is the noise precision.  We assume that $\bar{S}_\ell \geq 0$ and $\bar{S}_\ell \leq S_{\max}$
for some upper bound $S_{\max}$. }

Now define the quantities
\begin{align}  \label{eq:pq0}
\begin{split}
    \qbf^0_\ell &:= \zbf^0_\ell, \quad
    \pbf^0_\ell := \Vbf_\ell \qbf^0_\ell = \Vbf_\ell \zbf^0_\ell \quad \ell=0,2,\ldots,L \\
    \qbf^0_\ell &:= \Vbf_\ell\tran\zbf^0_\ell, \quad \pbf^0_\ell := \zbf^0_\ell = \Vbf_\ell \qbf^0_\ell, \quad \ell=1,3,\ldots,\Lm1,
\end{split}
\end{align}
which represent the true vectors $\zbf^0_\ell$ and their transforms.
For $\ell=0,2,\ldots,L-2$, we next define  the vectors:
\begin{subequations} \label{eq:pqdef}
\begin{align}
    &\qbfhat^{\pm}_{k\ell} = \zbfhat^{\pm}_{k\ell}, \quad
    \qbf^{\pm}_{k\ell} = \rbf_{k\ell}^\pm - \zbf^0_\ell, \label{eq:qdefeven} \\
    &\pbfhat^{\pm}_{k,\lp1} = \zbfhat^{\pm}_{k,\lp1}, \quad
    \pbf^{\pm}_{k,\lp1} = \rbf_{k,\lp1}^{\pm} - \zbf^0_{\lp1}, \label{eq:pdefodd}  \\
    &\qbfhat^{\pm}_{k,\lp1} = \Vbf_{\lp1}\tran\pbfhat^{\pm}_{k,\lp1}, \quad
    \qbf^{\pm}_{k,\lp1} = \Vbf_{\lp1}\tran\pbf^{\pm}_{k,\lp1} \label{eq:qdefodd} \\
    &\pbfhat^{\pm}_{k\ell} = \Vbf_\ell\qbfhat^{\pm}_{k\ell}, \quad
    \pbf^{\pm}_{k\ell} = \Vbf_\ell\qbf^{\pm}_{k\ell},  \label{eq:pdefeven}
\end{align}
\end{subequations}
The vectors $\qbfhat^{\pm}_{k\ell}$ and $\pbfhat^{\pm}_{k\ell}$ represent the estimates
of $\qbf^0_\ell$ and $\pbf^0_\ell$.
Also, the vectors $\qbf^{\pm}_{k\ell}$ and $\pbf^{\pm}_{k\ell}$
are the differences $\rbf_{k\ell}^{\pm}-\zbf^0_\ell$ or their transforms.  These
represent errors on the \emph{inputs} $\rbf_{k\ell}^\pm$ to the estimation functions
$\gbf^{\pm}_\ell(\cdot)$.

\begin{theorem} \label{thm:mapse}  Under the above assumptions,
for any fixed iteration $k$ and $\ell=1,\ldots,\Lm1$,
the components of
$\pbf^0_{\lm1}$, $\qbf^0_{\ell}$, $\pbf_{k,\lm1}^+$, $\qbf_{k\ell}^\pm$, $\qbfhat^+_{k\ell}$,
almost surely empirically converge jointly with limits,
\begin{align}
   \MoveEqLeft \lim_{N \arr \infty} \left\{
   (p^0_{\lm1,n},p^+_{k,\lm1,n},q^0_{\ell,n},q^-_{k\ell,n},q^+_{k\ell,n},\hat{q}^+_{k\ell,n}) \right\} \nonumber \\
        &\PLeq
        (P^0_{\lm1},P^+_{k,\lm1},Q^0_{\ell},Q^-_{k\ell}, Q^+_{k\ell},\hat{Q}_{k\ell}),
    \label{eq:PQplim1}
\end{align}
where the variables
$P^0_{\lm1}$, $P_{k\lm1}^+$ and $Q_{k\ell}^-$
are zero-mean jointly Gaussian random variables with
\begin{align} \label{eq:PQpcorr1}
\begin{split}
    &\Cov(P^0_{\lm1},P_{k,\lm1}^+) = \Kbf_{k,\lm1}^+, \quad \Exp(Q_{k\ell}^-)^2 = \tau_{k\ell}^-, \nonumber \\
    &\Exp(P_{k,\lm1}^+Q_{k\ell}^-)  = 0,  \quad \Exp(P^0_{\lm1}Q_{k\ell}^-)  = 0,
\end{split}
\end{align}
for parameters $\Kbf_{k,\lm1}^+$ and $\tau_{k\ell}^-$.
The identical result holds for $\ell=0$ with the variables $\pbf_{k,\lm1}^+$ and $P_{k,\lm1}^+$ removed.  Also, a similar result holds for the variables
$\pbf^0_{\lm1}$, $\pbf_{\kp1,\lm1}^+$, $\pbf_{k,\lm1}^+$,$\qbf_{\kp1,\ell}^-$.
\end{theorem}

\iftoggle{conference}{The full paper \cite{pandit2019asymptotics-arxiv} provides a precise and simple description
of the limiting random variables on the right hand side of \eqref{eq:PQplim1}.
In addition, the parameters
$\Kbf^{\pm}_{k\ell}$ and $\tau_{k\ell}^{\pm}$ can be computed by deterministic recursive formulae, thus
representing a state evolution (SE) for the MAP ML-VAMP system.  }{
Appendix~\ref{sec:mapsepf} states and proves the complete result.
The complete results provides a precise and simple description
of all the limiting random variables on the right hand side of \eqref{eq:PQplim1}.
In particular, all the random variables are either Gaussian or the outputs of nonlinear functions
of Gaussian.  In addition, the parameters of the Gaussian random variables 
such as $\Kbf^{\pm}_{k\ell}$ and $\tau_{k\ell}^{\pm}$ 
are given by a deterministic recursive algorithm (Algorithm~\ref{algo:gen_se}).
The recursive updates thus
represent a state evolution (SE) for the MAP ML-VAMP system.
}
In the case of MMSE estimation functions,
the SE equations reduce to those of \cite{fletcher2017inference}.  

The importance of this limiting model is that we can
compute several important performance metrics of the ML-VAMP system.  
\iftoggle{conference}{
For example, for a non-linear layer,
$\ell=0,2,\ldots,L$, it is shown in the full paper \cite{pandit2019asymptotics-arxiv}
that the asymptotic mean-squared error (MSE) is given by,
\[
    \lim_{N \arr \infty} \frac{1}{N} \|\zbf^0_\ell - \zbfhat^+_{k\ell}\|^2
    = \Exp(Q^0_\ell - \hat{Q}^+_{k\ell})^2,
\]
where the expectation can be computed from the model from the random variables in \eqref{eq:PQplim1}.}{
For example, let 
$\ell=0,2,\ldots,L$ be the index of a nonlinear layer. 
Then, the asymptotic mean-squared error (MSE) is given by,
\begin{align*}
    \MoveEqLeft \lim_{N \arr \infty} \frac{1}{N} \|\zbf^0_\ell - \zbfhat^+_{k\ell}\|^2 \nonumber \\
    &\stackrel{(a)}{=}
    \lim_{N \arr \infty} \frac{1}{N} \|\qbf^0_\ell - \qbfhat^+_{k\ell}\|^2 
    \stackrel{(b)}{=} \Exp(Q^0_\ell - \hat{Q}^+_{k\ell})^2,
\end{align*}
where (a) follows from the definitions in \eqref{eq:pq0} and \eqref{eq:pqdef};
and (b) follows from the definition of empirical convergence.  
The expectation $\Exp(Q^0_\ell - \hat{Q}^+_{k\ell})^2$ can then 
be computed from the model from the random variables in \eqref{eq:PQplim1}.}
In this way, we see that MAP ML-VAMP provides a computationally tractable method for computing
critical points of the MAP objective with precise predictions on its performance.

\section{Numerical Simulations} \label{sec:sim}
To validate the MAP ML-VAMP algorithm and the LSL analysis, we
simulate the method in a random synthetic network  similar to
\cite{fletcher2017inference}.  Details are given in Appendix~\ref{sec:simdetails}.
Specifically, we consider a network
with $N_0 = 20$ inputs and two hidden stages with 100 and 500 units with ReLU activations.
The number of outputs is $N_y$ is varied.  In the final layer, AWGN noise is added at an SNR of 20~dB.
The weight matrices have Gaussian i.i.d. components and the biases $b_\ell$
are selected so that the ReLU outputs are non-zero, on average, for 40\% of the
samples.  For each value of $N_y$, we generate 40 random instances of the network
and compute (a) the MAP estimate using the Adam optimizer \cite{kingma2014adam} in Tensorflow;
(b) the estimate from MAP ML-VAMP; and (c) the MSE for MAP ML-VAMP predicted by the state evolution.
Fig.~\ref{fig:randmlp_sim} shows the median normalized MSE,
$10 \log_{10}(\|\zbf^0_\ell-\zbfhat^+_{k\ell}\|^2/\|\zbf^0_\ell\|^2)$ for the input variable
($\ell=0$) for the three methods.  We see that for $N_y\geq 100$, the actual performance of MAP ML-VAMP
matches the SE closely as well as the performance of MAP estimation via a generic solver.
For $N_y<100$, the match is still close, but there is a small discrepancy, likely due to the relatively
small size of the problem.  Also, for small $N_y$, MAP ML-VAMP appears to achieve a slightly better performance
than the Adam optimizer.  Since both are optimizing the same objective, the difference is likely due to the
ML-VAMP finding better local minima.

To demonstrate that MAP ML-VAMP can also work on a
simple non-random dataset,
Fig.~\ref{fig:mnist_inpaint} shows
samples of reconstructions results for inpainting
for MNIST digits.  A VAE \cite{kingma2013auto} is
used to train a generative model.
The MAP ML-VAMP reconstruction
obtains similar results as MAP inference using the
Adam optimizer, although sometimes different local
minima are found. The main benefit is that
MAP ML-VAMP can be rigorously analyzed.
Details are in the full paper
\cite{pandit2019asymptotics-arxiv}.

\begin{figure}
\centering
\includegraphics[width=0.8\columnwidth]{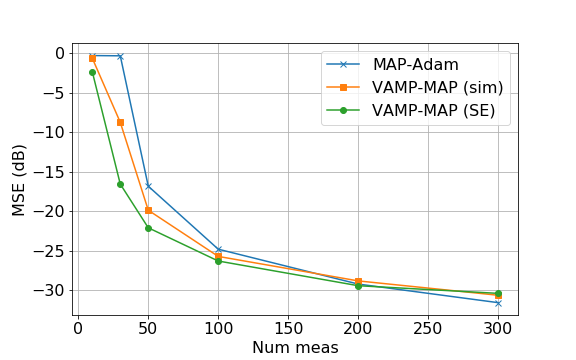}
\caption{Normalized MSE for a random multi-layer network for (a) MAP inference
computed by Adam optimizer; (b) MAP inference from ML-VAMP; (c) State evolution prediction.}
\label{fig:randmlp_sim}
\end{figure}

\begin{figure}
\centering
\includegraphics[width=0.8\columnwidth]{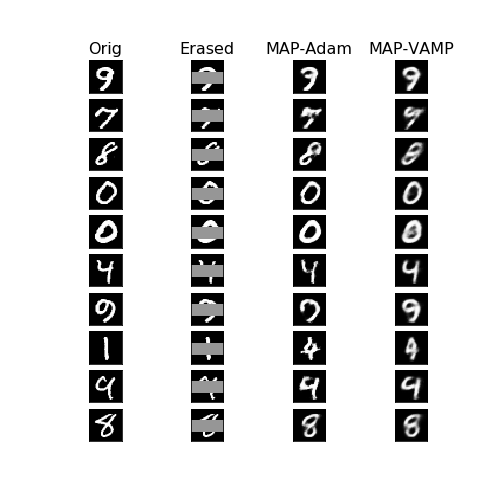}
\vspace{-0.5cm}
\caption{MNIST inpainting where the rows 10-20 of
the 28 $\times$ 28 digits are erased. }
\label{fig:mnist_inpaint}
\end{figure}

\section*{Conclusions}
\iftoggle{conference}{}{MAP inference combined with deep generative priors provides a powerful tool for
complex inverse problems.  Rigorous analysis of these methods has been difficult.}
ML-VAMP with MAP estimation provides a computationally tractable method for performing the MAP
inference with performance that can be rigorously and precisely characterized in a certain large
system limit.  The approach thus offers a new and potentially powerful approach for understanding
and improving deep network-based inference.

\iftoggle{conference}{}{

\appendices

\input{appendices}

}

\iftoggle{conference}{
\section*{Acknowledgements}
AKF was supported by NSF grants 1254204 and 1738286, and ONR N00014-15-1-2677. SR was supported by NSF Grants 1116589, 1302336, and 1547332, NIST award 70NANB17H166,
SRC, and NYU WIRELESS.
}
{}

\bibliographystyle{IEEEtran}
\bibliography{bibl_vamp}

\end{document}

%% file: macros_parthe.tex
\newcommand{\mc}{\mathcal}

%% file: abstract.tex
\begin{abstract}
Deep generative priors are a powerful tool for reconstruction problems
with complex data such as images and text.  Inverse problems using such models
require solving an inference problem of estimating the input and hidden units of the
multi-layer network from its output.  Maximum a priori (MAP) estimation is a
widely-used inference method as it is straightforward to implement, and
has been successful in practice.  However, rigorous
analysis of MAP inference in multi-layer networks is difficult.
This work considers a recently-developed method, multi-layer
vector approximate message passing (ML-VAMP), to study MAP inference in deep networks.
It is shown that the mean squared error of the
ML-VAMP estimate can be exactly and rigorously characterized
in a certain high-dimensional random limit.  The proposed method thus
provides a tractable method for MAP inference with exact performance guarantees.
\end{abstract}

%% file: appendices.tex

\section{Proof of Theorem~\ref{thm:mapfix}} \label{sec:mapfixpf}

The linear equalities in \eqref{eq:defines} can be rewritten as,
\begin{subequations}  \label{eq:defr}
\begin{align}
    \rbf_{k\ell}^+ &= \zbfhat^+_{k\ell} + \frac{1}{\alpha^{-}_{k\ell}}\sbf^{+}_{k\ell}
          \\
    \rbf_{k+1,\ell}^- &= \zbfhat^-_{k\ell} - \frac{1}{\alpha^{+}_{k\ell}}\sbf^{-}_{\kp1,\ell} 
\end{align}
\end{subequations}
Substituting \eqref{eq:defr}  in lines \ref{line:rp} and \ref{line:rn} of Algorithm \ref{algo:ml-vamp} give the updates \eqref{eq:admmsp} and \eqref{eq:admmsn} in Theorem \ref{thm:mapfix}. It remains to show that the optimization problem in updates
\eqref{eq:admmHp} and \eqref{eq:admmHn} is equivalent to
\eqref{eq:mapEnergy_l}. It suffices to show that the terms dependent on $(z_{\ell-1}^-,z^{+}_\ell)$ in both the objective functions $J_\ell$ from \eqref{eq:mapEnergy_l} and $\mc L_\ell$ from \eqref{eq:admmHp} and \eqref{eq:admmHn} are identical. This follows immediately on substituting \eqref{eq:defr} in \eqref{eq:defineEnergy_l}.

It now suffices to show that any fixed point of Algorithm~\ref{algo:ml-vamp} is a critical point of the augmented Lagrangian in \eqref{eq:Lagdef}.
Since we are looking only at fixed points, we can drop the dependence
on the iteration $k$.  So, for example, we can write $\rbf_\ell^+$ for $\rbf_{k\ell}^+$.
To show that $\zbfhat^+_\ell,\zbfhat^-_\ell$ are critical points of the
constrained optimization~\eqref{eq:Fmincon}, we need to show that
there exists dual parameters $\sbf_\ell$ such that for all $\ell=0,\ldots,\Lm1$,
\beq \label{eq:zcon}
    \zbfhat^+_\ell=\zbfhat^-_\ell,
\eeq
\beq \label{eq:Laggrad}
    \frac{\partial \mc L(\zbfhat^+,\zbfhat^-,\sbf)}{\partial \zbf^+_\ell}=0, \quad
    \frac{\partial \mathcal{L} (\zbfhat^+,\zbfhat^-,\sbf)}{\partial \zbf^-_\ell}=0,
\eeq
where $\mc L(\cdot)$ is the Lagrangian in \eqref{eq:Lagdef}.

We first prove \eqref{eq:zcon} whereby primal feasibility is satisfied.
At any fixed point of \eqref{eq:gamupdate}, we have
\[
    \eta_\ell=\gamma^+_\ell + \gamma^-_\ell = \frac{\gamma^+_\ell}{\alpha_\ell^-} =
     \frac{\gamma_\ell^-}{\alpha_\ell^+}.
\]
Therefore,
\beq\label{eq:alphaone}
    \alpha_\ell^- = \frac{\gamma^+_\ell}{\gamma_\ell^++\gamma_\ell^-}
        = 1- \frac{\gamma^-_\ell}{\gamma_\ell^++\gamma_\ell^-} = 1-\alpha^+_\ell.
\eeq
Now, from line~\ref{line:rp} in Algorithm~\ref{algo:ml-vamp},
\begin{align}
     \zbfhat^+_\ell &= (1-\alpha^+_\ell)\rbf^+_\ell + \alpha^+_\ell\rbf^-_\ell \nonumber \\
       & = \alpha^-_\ell\rbf^+_\ell + \alpha^+_\ell\rbf^-_\ell, \label{eq:zhatfix1}
\end{align}
where the last step used \eqref{eq:alphaone}.  Similarly, from line~\ref{line:rn},
\beq
    \zbfhat^-_\ell = \alpha^-_\ell\rbf^+_\ell + \alpha^+_\ell\rbf^-_\ell. \label{eq:zhatfix2}
\eeq
Equations \eqref{eq:zhatfix1} and \eqref{eq:zhatfix2} prove \eqref{eq:zcon}.
In the sequel, we will let $\zbfhat_\ell$ denote $\zbfhat^+_\ell$ and $\zbfhat^-_\ell$
since they are equal. As a consequence of the primal feasibility $\zbfhat^+_\ell=\zbfhat^-_\ell$, observe that
\beq
\sbf_\ell^+-\sbf_\ell^- = (\alpha_\ell^++\alpha_\ell^-)\zbfhat_\ell - \alpha^+_\ell\rbf_\ell^--\alpha^-_\ell\rbf_\ell^+=0,
\eeq
where we have used \eqref{eq:alphaone} and \eqref{eq:zhatfix1}. Define $\sbf:=\sbf^+=\sbf^-$, by virtue of the equality shown above.

Having shown the equivalence of Algorithm \ref{algo:ml-vamp} and the iterative updates in the statement of the theorem, we can say that there exists a one-to-one linear mapping between their fixed points $\{\zbfhat,\rbf^{+},\rbf^{-}\}$ (from Algorithm \ref{algo:ml-vamp}) and $\{\zbfhat,\sbf\}$ (from Theorem \ref{thm:mapfix}). Now to show \eqref{eq:Laggrad} it suffices to show that $\sbf_\ell$ is a valid dual parameter for which the following stationarity conditions hold,
\begin{align} \label{eq:Laggrad_l}
    \frac{\partial \mc L_\ell(\zbf^-_{\ell-1},\zbf^+_{\ell};\zbfhat^+_{\ell-1},\zbfhat^-_{\ell},\sbf_{\ell-1},\sbf_{\ell})}{\partial \zbf^-_{\ell-1}}\Bigg{\rvert}_{(\zbfhat_{\ell-1}^-,\zbfhat_{\ell}^+)}\owns\ &\ {\bf 0}, \\
    \frac{\partial \mc L_\ell(\zbf^-_{\ell-1},\zbf^+_{\ell};\zbfhat^+_{\ell-1},\zbfhat^-_{\ell},\sbf_{\ell-1},\sbf_{\ell})}{\partial \zbf^+_\ell}\Bigg{\rvert}_{(\zbfhat_{\ell-1}^-,\zbfhat_{\ell}^+)}\owns\ &\ {\bf 0}.
\end{align}
Indeed the above conditions are the stationarity conditions of the optimization problem in \eqref{eq:admmHp} and \eqref{eq:admmHn}.

\section{Empirical Convergence of Random Variables} \label{sec:empirical}

We follow the framework of Bayati and Montanari \citep{BayatiM:11}, which models
various sequences as deterministic, but with components converging empirically
to a distribution.  We start with a brief review of useful definitions.
Let $\xbf = (\xbf_1,\ldots,\xbf_N)$ be a block vector with components $\xbf_n \in \R^r$
for some $r$.  Thus, the vector $\xbf$
is a vector with dimension $rN$.
Given any function $g:\R^r \arr \R^s$, we define the
\emph{componentwise extension} of $g(\cdot)$ as the function,
\beq \label{eq:gcomp}
    \gbf(\xbf) := (g(\xbf_1),\ldots,g(\xbf_N)) \in \R^{Ns}.
\eeq
That is,
$\gbf(\cdot)$
applies the function
$g(\cdot)$
on each $r$-dimensional component.
Similarly,
we say $\gbf(\xbf)$ \emph{acts componentwise} on $\xbf$ whenever it is of the form \eqref{eq:gcomp}
for some function $g(\cdot)$.

Next consider a sequence of block vectors of growing dimension,
\[
    \xbf(N) = (\xbf_1(N),\ldots,\xbf_N(N)),
\qquad
N=1,\,2,\,\ldots,
\]
where each component $\xbf_n(N) \in \R^r$.
In this case, we will say that
$\xbf(N)$ is a \emph{block vector sequence that scales with $N$
under blocks $\xbf_n(N) \in \R^r$.}
When $r=1$, so that the blocks are scalar, we will simply say that
$\xbf(N)$ is a \emph{vector sequence that scales with $N$}.
Such vector sequences can be deterministic or random.
In most cases, we will omit the notational dependence on $N$ and simply write $\xbf$.

Now, given $p \geq 1$,
a function $f:\R^r \arr \R^s$ is called \emph{pseudo-Lipschitz continuous of order $p$},
if there exists a constant $C > 0$ such that for all $\xbf_1,\xbf_2 \in\R^r$,
\[
    \| f(\xbf_1)- f(\xbf_2) \| \leq C\|\xbf_1-\xbf_2\|\left[ 1 + \|\xbf_1\|^{p-1}
    + \|\xbf_2\|^{p-1} \right].
\]
Observe that in the case $p=1$, pseudo-Lipschitz continuity reduces to
the standard Lipschitz continuity.
Given $p \geq 1$, we will say that the block vector sequence $\xbf=\xbf(N)$
\emph{converges empirically with $p$-th order moments} if there exists a random variable
$X \in \R^r$ such that
\begin{enumerate}[(i)]
\item $\Exp\|X\|_p^p < \infty$; and
\item for any $f : \R^r \arr \R$ that is pseudo-Lipschitz continuous of order $p$,
\beq \label{eq:PLp-empirical}
    \lim_{N \arr \infty} \frac{1}{N} \sum_{n=1}^N f(\xbf_n(N)) = \Exp\left[ f(X) \right].
\eeq
\end{enumerate}
In \eqref{eq:PLp-empirical}, we have
the empirical mean of the components $f(\xbf_n(N))$
of the componentwise extension $\fbf(\xbf(N))$
converging to the expectation $\Exp[ f(X) ]$.
In this case, with some abuse of notation, we will write
\beq \label{eq:plLim}
    \lim_{N \arr \infty} \left\{ \xbf_n \right\} \stackrel{PL(p)}{=} X,
\eeq
where, as usual, we have omitted the dependence on $N$ in $\xbf_n(N)$.
Importantly, empirical convergence can de defined on deterministic vector sequences,
with no need for a probability space.  If $\xbf=\xbf(N)$ is a random vector sequence,
we will often require that the limit \eqref{eq:plLim} holds almost surely.

We conclude with one final definition.
Let $\phibf(\rbf,\gamma)$ be a function on $\rbf \in \R^s$ and $\gamma \in \R$.
We say that $\phibf(\rbf,\gamma)$ is \emph{uniformly Lipschitz continuous} in $\rbf$
at $\gamma=\gammabar$ if there exists constants
$L_1$ and $L_2 \geq 0$ and an open neighborhood $U$ of $\gammabar$, such that
\beq \label{eq:unifLip1}
    \|\phibf(\rbf_1,\gamma)-\phibf(\rbf_2,\gamma)\| \leq L_1\|\rbf_1-\rbf_2\|,
\eeq
for all $\rbf_1,\rbf_2 \in \R^s$ and $\gamma \in U$; and
\beq \label{eq:unifLip2}
    \|\phibf(\rbf,\gamma_1)-\phibf(\rbf,\gamma_2)\| \leq L_2\left(1+\|\rbf\|\right)|\gamma_1-\gamma_2|,
\eeq
for all $\rbf \in \R^s$ and $\gamma_1,\gamma_2 \in U$.

\section{Large System Limit: Model Details} \label{sec:lsl}

In addition to the assumptions in Section~\ref{sec:seevo},
we describe a few more technical assumptions.
First, we need that the activation functions
$\phibf_\ell(z_{\lm1},\xi_\ell)$ in \eqref{eq:nnnonlintrue} act \emph{componentwise} meaning
that,
\beq \label{eq:phicomp}
 \left[ \phibf_{\ell}(\zbf_{\lm1}, \xibf_\ell) \right]_n
    = \phi_\ell(z_{\lm1,n},\xi_{\ell,n})
\eeq
for some scalar-valued function $\phi_\ell(\cdot)$ for all components $n$. That is, for a nonlinear layer
$\ell=2,4,\ldots,L$, each output $z^0_{\ell,n}$ depends only on the corresponding input component
$z^0_{\lm1,n}$.  Standard activations such as ReLU or sigmoid would satisfy this property.
In addition, we require that the activation function components $\phi_\ell(\cdot)$
are pseudo-Lipschitz continuous of order two.

Next, we
need certain assumptions on the estimation functions $\gbf_\ell^{\pm}(\cdot)$.
For the estimation functions corresponding to the  nonlinear layers,
$\ell=2,4,\ldots,L-2$, we assume that for each parameter $\theta_{k\ell}^-$, the function
$\gbf_\ell^+(\rbf_{\lm1}^+,\rbf_{\ell}^-,\theta_{\ell}^-)$ is Lipschitz
continuous in $(\rbf^+_{\lm1},\rbf^-_\ell)$ and $\gbf^+_\ell(\cdot)$ acts
\emph{componentwise} in that,
\beq \label{eq:gprow}
    \zbfhat^+_\ell = \gbf_\ell^+(\rbf^+_{\lm1},\rbf^-_\ell,\theta_\ell^+)
    \Leftrightarrow
    \zhat^+_{\ell,i} = g_\ell^+(r^+_{\lm1,i},r^-_{\ell,i},\theta_\ell^+),
\eeq
for $i=1,\ldots,N_\ell$.
for some scalar-valued function $g_\ell^+(\cdot)$. Thus, each element
$\zhat^+_{\ell,i}$ of the output vector $\zbfhat^+_\ell$ depends only the corresponding elements of the inputs
$r^+_{\lm1,i}$ and $r^-_{\ell,i}$.  We make a similar assumption on the first estimation function
$\gbf^+_0(\cdot)$ as well as the reverse functions $\gbf^-_\ell(\cdot)$ for $\ell=2,4,\ldots,L$ and define $g_0^+(\cdot)$ and $g_\ell^-(\cdot)$ in a similar manner.

Note that for the linear layers $\ell=1,3,\ldots,\Lm1$, we assume the MAP denoiser \eqref{eq:gmap}.
For the linear layer, this is identical to the MMSE denoiser and the estimation functions can be written as,
\begin{subequations}\label{eq:glintrans}
\begin{align}
  \MoveEqLeft \gbf^+_\ell(\rbf^+_{\lm1},\rbf^-_{\lm1},\theta^+_{\ell}) \nonumber \\
  &=
  \Vbf_\ell{\Gbf}^+_\ell(\Vbf_{\lm1}\rbf^+_{\lm1},\Vbf_{\ell}\tran\rbf^-_{\lm1},
  \bar{\sbf}_\ell,\bar{\bbf}_\ell,\theta^-_{\ell}) \\
  \MoveEqLeft \gbf^-_\ell(\rbf^+_{\lm1},\rbf^-_{\lm1},\theta^+_{\ell}) \nonumber \\
   &:=
  \Vbf_{\lm1}\tran{\Gbf}^+_\ell(\Vbf_{\lm1}\rbf^+_{\lm1},\Vbf_{\ell}\tran\rbf^-_{\lm1},
  \bar{\sbf}_\ell,\bar{\bbf}_\ell,   \theta^-_{\ell}),
\end{align}
\end{subequations}
where, for each parameter value $\theta^{\pm}_\ell$, the functions ${\Gbf}^{\pm}_\ell(\cdot)$ are Lipschitz continuous in
$(\rbf^+_{\lm1},\rbf^-_\ell,\bar{\sbf}_\ell)$ and are componentwise extensions of ${G}_\ell^{\pm}$ defined as,
\begin{align}
    \MoveEqLeft \begin{bmatrix}
        {G}_\ell^-(\bar{u}_{\lm1},\bar{u}_\ell,s_\ell,\bar{b}_\ell,\gamma_{\lm1}^+,\gamma_{\ell}^-) \\
        {G}_\ell^+(\bar{u}_{\lm1},\bar{u}_\ell,s_\ell,\bar{b}_\ell,\gamma_{\lm1}^+,\gamma_{\ell}^-)
    \end{bmatrix}
    \nonumber \\
       & =  \begin{bmatrix}
        \gamma_{\lm1}^+ + \nu_\ell s_\ell^2 & -\nu_\ell s_\ell \\
        -\nu_\ell s_\ell & \gamma_\ell^- +\nu_\ell
        \end{bmatrix}^{-1}
        \begin{bmatrix}
        \gamma_{\lm1}^+ \bar{u}_{\lm1} - \nu_\ell s_\ell \bar{b}_\ell \\
        \gamma_{\ell}^- \bar{u}_{\ell} + \nu_\ell \bar{b}_\ell
        \end{bmatrix}, \label{eq:gdeflin}
\end{align}
We call the functions ${\Gbf}^{\pm}_\ell(\cdot)$, the \emph{transformed denoising
functions}. We refer the reader to the appendices
of \cite{fletcher2017inference} for a detailed derivation of $\Gbf^{\pm}_\ell$.
We now need two further technical assumptions.

\begin{algorithm}[t]
\caption{Transformed ML-VAMP Recursion}
\begin{algorithmic}[1]  \label{algo:trans}
\STATE{// Initialization }
\STATE{Initial vectors $\wbf_\ell$, $\qbf^0_0$, $\qbf_{0\ell}^-$}
    \label{line:init_gen}
\STATE{$\qbf^0_0 = \fbf^0_0(\wbf_0), \quad \pbf^0_0 = \Vbf_0\qbf^0_0$} \label{line:q00init_gen}
\FOR{$\ell=1,\ldots,\Lm1$}
    \STATE{$\qbf^0_\ell = \fbf^0_\ell(\pbf^0_{\lm1},\wbf_\ell, \Lambda_{01}^-)$ }
    \label{line:q0init_gen}
    \STATE{$\pbf^0_\ell = \Vbf_\ell\qbf^0_\ell$ }  \label{line:p0init_gen}
\ENDFOR
\STATE{}

\FOR{$k=0,1,\dots,N_{\rm it}-1$}

    \STATE{// Forward Pass }
    \STATE{$\qbfhat^+_{k0} = \hbf^+_0(\qbf^-_{k\ell},\wbf_\ell,\theta^+_{k0})$ }
    \STATE{$\alpha^+_{k0} = \bkt{\partial
        \hbf^+_0(\qbf^-_{k0},\wbf_\ell,\theta^+_{k0})/\partial \qbf^-_{k0}}$}
        \label{line:alphagenp0}
    \STATE{$\Lambda_{k0}^+ = (\alpha_{k0}^+,\theta_{k0}^+)$}
    \STATE{$\qbf^+_{k0} = \fbf^+_0(\qbf^-_{k0},\wbf_\ell,\Lambda^+_{k0})$ }
    \STATE{$\pbf^+_{k0} = \Vbf_0 \qbf^+_{k0}$}
    \FOR{$\ell=1,\ldots,\Lm1$}
        \STATE{$\qbfhat^+_{k\ell} = \hbf^+_\ell(\pbf^0_{\lm1},\pbf^+_{k,\lm1},
            \qbf^-_{k\ell},\wbf_\ell,\theta^+_{k\ell})$ }
        \STATE{$\alpha^+_{k\ell} = \bkt{\partial
            \hbf^+_\ell(\pbf^+_{k,\lm1},\qbf^-_{k\ell},\wbf_\ell,\theta^+_{k\ell})/\partial \qbf^-_{k\ell}}$}
            \label{line:alphagenp}
        \STATE{$\Lambda_{k\ell}^+ = (\alpha_{k\ell}^+,\theta_{k\ell}^+)$}
        \STATE{$\qbf^+_{k\ell} = \fbf^+_\ell(\pbf^+_{k,\lm1},\qbf^-_{k\ell},\wbf_\ell,
            \Lambda^+_{k\ell})$ }
        \STATE{$\pbf^+_{k\ell} = \Vbf_\ell \qbf^+_{k\ell}$}
    \ENDFOR
    \STATE{}

    \STATE{// Reverse Pass }
    \STATE{$\pbfhat^-_{\kp1,\Lm1} =
        \hbf^-_{L}(\pbf^+_{k,\Lm1},\wbf_{L},\theta^-_{kL})$ }
    \STATE{$\alpha^-_{k,\Lm1} = \bkt{\partial
        \hbf^-_{L}(\pbf^+_{k,\Lm1},\wbf_{L},
        \theta^-_{kL})/\partial \pbf^+_{k,\Lm1}}$}
        \label{line:alphagennL}
    \STATE{$\Lambda_{kL}^- = (\alpha_{k,\Lm1}^-,\theta_{kL}^-)$}
    \STATE{$\pbf^-_{\kp1,\Lm1} =
        \fbf^-_{L}(\pbf^+_{k,\Lm1},\wbf_{L},\Lambda^-_{kL})$ }
    \FOR{$\ell=L-2,\ldots,0$}
        \STATE{$\pbfhat^-_{\kp1,\lm1} = \hbf^-_{\lp1}(\pbf^+_{k\ell},\qbf^-_{k,\lp1},\wbf_{\lp1},\theta^-_{k,\lp1})$}
        \STATE{$\alpha^-_{k\ell} =
        \bkt{\partial \hbf^-_{\lp1}(\pbf^+_{k\ell},\cdots)/\partial \pbf^+_{k\ell}}$}
            \label{line:alphagenn}
        \STATE{$\Lambda_{k,\lp1}^- = (\alpha_{k,\lp1}^-,\theta_{k,\lp1}^-)$}
        \STATE{$\pbf^-_{\kp1,\ell} =
            \fbf^-_{\lp1}(\pbf^+_{k\ell},\qbf^-_{k,\lp1},\wbf_{\lp1},\Lambda^-_{k,\lp1})$ }
        \STATE{$\qbf^-_{\kp1,\ell} = \Vbf_\ell\tran \pbf^-_{\kp1,\ell}$}
    \ENDFOR

\ENDFOR
\end{algorithmic}
\end{algorithm}

\begin{algorithm}[t]
\caption{State Evolution for ML-VAMP}
\begin{algorithmic}[1]  \label{algo:gen_se}

\REQUIRE{Vector update component functions $f^0_\ell(\cdot)$ and $f^\pm_{k\ell}(\cdot)$}

\STATE{}
\STATE{// Initial pass}
\STATE{Initial random variables:  $W_\ell$, $Q_{0\ell}^-$, $\ell=0,\ldots,\Lm1$}
    \label{line:qinit_se_gen}
\STATE{$Q^0_0 = f^0_0(W_0)$} \label{line:q0init_se_gen}
\STATE{$P^0_0 \sim \Norm(0,\tau^0_0)$, $\tau^0_0 = \Exp(Q^0_0)^2$} \label{line:p0init_se_gen}
\FOR{$\ell=1,\ldots,\Lm1$}
    \STATE{$Q^0_\ell=f^0_\ell(P^0_{\lm1},W_\ell)$}
    \STATE{$P^0_\ell = \Norm(0,\tau^0_\ell)$,
            $\tau^0_\ell = \Exp(Q^0_\ell)^2$} \label{line:pinit_se_gen}
\ENDFOR
\STATE{}

\FOR{$k=0,1,\dots$}
    \STATE{// Forward Pass }
    \STATE{$\hat{Q}^+_{k0} = h^+_0(Q_{k0}^-,W_0,\theta^+_{k0}))$} \label{line:qhat0_se_gen}
    \STATE{$\alphabar_{k0}^+ = \Exp(\partial h^+_0(Q_{k0}^-,W_0,\theta^+_{k0})/\partial Q_{k0}^-)$}
    \STATE{$\Lambdabar_{k0}^+ = (\alphabar^+_{k0},\theta_{k0}^+)$}
    \STATE{$Q_{k0}^+ = f^+_{k0}(Q_{k0}^-,W_0,\Lambdabar^+_{k0})$}  \label{line:q0_se_gen}
    \STATE{$(P^0_0,P_{k0}^+) = \Norm(\zero,\Kbf_{k0}^+)$,
        $\Kbf_{k0}^+ = \Cov(Q^0_0,Q_{k0}^+)$} \label{line:p0_se_gen}
    \FOR{$\ell=1,\ldots,L-1$}
        \STATE{$\hat{Q}^+_{k\ell} = h^+_\ell(P^0_{\lm1},P^+_{k,\lm1},Q_{k\ell}^-,
            W_\ell,\theta^+_{k\ell}))$} \label{line:qhat_se_gen}
        \STATE{$\alphabar_{k0}^+ = \Exp(\partial h^+_\ell(\ldots))/\partial Q_{k\ell}^-)$}
        \STATE{$\Lambdabar_{k\ell}^+ = (\alphabar^+_{k\ell},\theta_{k\ell}^+)$}
            \label{line:lamp_se_gen}
        \STATE{$Q_{k\ell}^+ = f^+_{k\ell}(P^0_{\lm1},P^+_{k,\lm1},Q_{k\ell}^-,W_\ell,\Lambdabar^+_{k\ell})$}
            \label{line:qp_se_gen}
        \STATE{$(P^0_\ell,P_{k\ell}^+) = \Norm(\zero,\Kbf_{k\ell}^+)$,
            $\Kbf_{k\ell}^+ = \Cov(Q^0_\ell,Q_{k\ell}^+) $}   \label{line:pp_se_gen}
    \ENDFOR
    \STATE{}

    \STATE{// Reverse Pass }
    \STATE{$\hat{P}_{\kp1,\Lm1}^- = h^-_{kL}(P^0_{\Lm1},P_{k,\Lm1}^+,W_L,\theta^-_{\kp1,L})$}
        \label{line:phatL_se_gen}
    \STATE{$\alphabar_{k,\Lm1}^- = \partial h^-_{kL}(\cdots)/  \partial P_{k,\Lm1}^+$}
    \STATE{$\Lambdabar_{k,\Lm1}^- = (\alphabar^-_{k,\Lm1},\theta_{k,\Lm1}^-)$}
    \STATE{$P_{\kp1,\Lm1}^- = f^-_{kL}(P^0_{\Lm1},P_{k,\Lm1}^+,W_L,\Lambdabar^-_{\kp1,L})$}  \label{line:pL_se_gen}
    \STATE{$\tau_{\kp1,\Lm1}^- = \Exp(P^-_{\kp1,\Lm1})^2$} \label{line:tauL_se_gen}
    \STATE{$Q_{\kp1,\Lm1}^- = \Norm(0,\tau_{\kp1,\Lm1}^-)$} \label{line:qL_se_gen}
    \FOR{$\ell=\Lm1,\ldots,1$}
        \STATE{$\hat{P}_{\kp1,\lm1}^- = h^-_{k\ell}(P^0_{\lm1},P_{k,\lm1}^+,W_\ell,\theta^-_{\kp1,\ell})$}
            \label{line:phatn_se_gen}
        \STATE{$\alphabar_{k,\lm1}^- = \Exp(\partial h^-_{k\ell}(\ldots)/ \partial P_{k,\Lm1}^+)$}
        \STATE{$\Lambdabar_{k\ell}^- = (\alphabar^-_{k,\lm1},\theta_{k,\ell}^-)$}
       \STATE{$P_{\kp1,\lm1}^- =
        f^-_{k\ell}(P^0_{\lm1},P^+_{k,\lm1},Q_{\kp1,\ell}^-,W_\ell,\Lambdabar^-_{k\ell})$}
            \label{line:pn_se_gen}
        \STATE{$\tau_{\kp1,\lm1}^- = \Exp(P_{\kp1,\lm1}^-)^2$} \label{line:taun_se_gen}
        \STATE{$Q_{\kp1,\lm1}^- = \Norm(0,\tau_{\kp1,\lm1}^-)$}   \label{line:qn_se_gen}
    \ENDFOR

\ENDFOR
\end{algorithmic}
\end{algorithm}

\section{Proof of Theorem \ref{thm:mapse}}  \label{sec:mapsepf}

\subsection{Transformed MLP}

The SE analysis of MMSE ML-VAMP in \cite{fletcher2017inference} proves a result on
a general class of multi-layer recursions, called Gen-ML.  To prove Theorem~\ref{thm:mapse},
we will show that ML-VAMP algorithm in Algorithm~\ref{algo:ml-vamp} is of the form
of an almost identical recursion with some minor changes in notation.
Theorem~\ref{thm:mapse} in this paper
will then follow from applying the general result in
\cite{fletcher2017inference}.  Since most of the proof is identical, we highlight only the
main differences.

Similar to \cite{fletcher2017inference}, we rewrite the MLP in \eqref{eq:nntrue} in
a certain transformed form.  To this end, define the disturbance vectors,
\begin{subequations} \label{eq:wdef}
\begin{align}
    \wbf_0 &:= \zbfhat^0_0, \quad
    \wbf_{\ell} := \xibf_\ell, \quad \ell =2,4,\ldots,L \\
    \wbf_{\ell} &= (\bar{\sbf}_{\ell},\bar{\bbf}_\ell,\bar{\xibf}_\ell),  \quad \ell =1,3,\ldots,\Lm1.
\end{align}
\end{subequations}
Also, define the scalar-valued functions,
\begin{subequations} \label{eq:f0ml}
\begin{align}
    \MoveEqLeft f^0_0(w_0) := w_0, \\
    \MoveEqLeft f^0_\ell(p^0_{\lm1},w_\ell) = f^0_\ell(p^0_{\lm1},\xi_\ell) := \phi_\ell(p^0_{\lm1},\xi_\ell), \nonumber \\
    & \ell=2,4,\ldots,L    \label{eq:f0nonlin} \\
    \MoveEqLeft f^0_\ell(p^0_{\lm1},w_\ell) = f^0_\ell(p^0_{\lm1},(\bar{s}_\ell,\bar{b}_\ell,\bar{\xi}_\ell))
        = \bar{s}_\ell p^0_\ell + \bar{b}_\ell + \bar{\xi}_\ell, \nonumber \\
    & \ell=1,3,\ldots,\Lm1.    \label{eq:f0lin}
\end{align}
\end{subequations}
Let $\fbf^0_\ell(\cdot)$ be their \emph{componentwise extension}, meaning that
\[
    \left[ \fbf^0_\ell(\pbf^0_{\lm1,n},\wbf_{\ell,n})\right]_n
    := f^0_\ell(p^0_{\lm1,n},w_{\ell,n}),
\]
so that $\fbf^0_\ell(\cdot)$ acts with the scalar-valued function $f^0_\ell(\cdot)$ on
each component of the vectors.
With this definition, it is shown in \cite{fletcher2017inference} that
the vectors $\pbf^0_\ell$ and $\qbf^0_\ell$ satisfy the recursions in
``Initialization" section of Algorithm~\ref{algo:trans}, the transformed algorithm.
This system of equations is represented diagrammatically in
the top panel of Fig.~\ref{fig:transformed}.
In comparison to Fig.~\ref{fig:nn_ml_vamp}, the transforms $\Wbf_\ell$ of the linear
 layers have been expanded using the SVD $\Wbf_{\ell} = \Vbf_{\ell}\Sigmabf_\ell\Vbf_{\lm1}$
 and inserting intermediate variables $\qbf^0_\ell$ and $\pbf^0_\ell$.
With the transformation, the  MLP \eqref{eq:nntrue} is equivalent to
a sequence of linear transforms by orthogonal matrices $\Vbf_\ell$ and non-linear componentwise
mappings $f^0_\ell(\cdot)$.

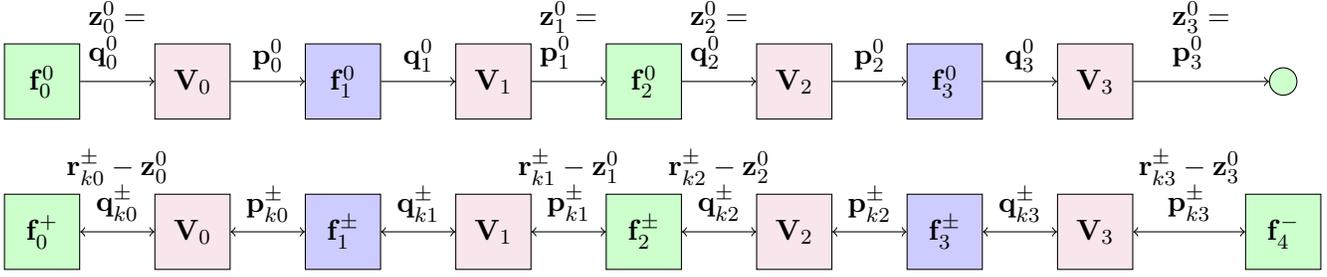
\begin{figure*}
\begin{tikzpicture}

    \tikzstyle{var}=[draw,circle,fill=green!20,node distance=2.5cm]
    \tikzstyle{linblock}=[draw,fill=blue!20, minimum size=1cm, node distance=2cm];
    \tikzstyle{nlblock}=[draw,fill=green!20, minimum size=1cm, node distance=2cm];
    \tikzstyle{vblock}=[draw,fill=purple!10, minimum size=1cm, node distance=2cm];

    \node [nlblock] (z0) {$\fbf^0_0$};
    \node [vblock, right of=z0] (V0) {$\Vbf_0$};
    \node [linblock, right of=V0] (S1) {$\fbf^0_1$};
    \node [vblock, right of=S1] (V1) {$\Vbf_1$};
    \node [nlblock, right of=V1] (phi2) {$\fbf^0_2$};
    \node [vblock, right of=phi2] (V2) {$\Vbf_2$};
    \node [linblock, right of=V2] (S3) {$\fbf^0_3$};
    \node [vblock, right of=S3] (V3) {$\Vbf_3$};
    \node [var, right of=V3] (phi4) {};

    \node [nlblock,below of=z0] (h0) {$\fbf_0^+$};
    \node [vblock, below of=V0] (V0b) {$\Vbf_0$};
    \node [linblock, below of=S1] (h1) {$\fbf_1^{\pm}$};
    \node [vblock, below of=V1] (V1b) {$\Vbf_1$};
    \node [nlblock, below of=phi2] (h2) {$\fbf_2^{\pm}$};
    \node [vblock, below of=V2] (V2b) {$\Vbf_2$};
    \node [linblock, below of=S3] (h3) {$\fbf_3^{\pm}$};
    \node [vblock, below of=V3] (V3b) {$\Vbf_3$};
    \node [nlblock, below of=phi4] (h4) {$\fbf_4^-$};

    \path[->] (z0) edge  node [above] {{\begin{tabular}{l} $\zbf^0_0=$ \\ $\qbf^0_0$ \end{tabular}}} (V0);
    \path[->] (V0) edge  node [above] {$\pbf^0_0$} (S1);
    \path[->] (S1) edge  node [above] {$\qbf_1^0$} (V1);
    \path[->] (V1) edge  node [above] {{\begin{tabular}{l} $\zbf^0_1=$ \\ $\pbf^0_1$ \end{tabular}}} (phi2);
    \path[->] (phi2) edge  node [above] {{\begin{tabular}{l} $\zbf^0_2=$ \\ $\qbf^0_2$ \end{tabular}}} (V2);
    \path[->] (V2) edge  node [above] {$\pbf^0_2$} (S3);
    \path[->] (S3) edge  node [above] {$\qbf^0_3$} (V3);
    \path[->] (V3) edge  node [above] {{\begin{tabular}{l} $\zbf^0_3=$ \\ $\pbf^0_3$ \end{tabular}}} (phi4);

    \path[<->] (h0) edge  node [above]
        {{\begin{tabular}{c} $\rbf^{\pm}_{k0}-\zbf^0_0$ \\ $\qbf^{\pm}_{k0}$ \end{tabular}}} (V0b);
    \path[<->] (V0b) edge  node [above] {$\pbf^{\pm}_{k0}$} (h1);
    \path[<->] (h1) edge  node [above] {$\qbf^{\pm}_{k1}$} (V1b);
    \path[<->] (V1b) edge  node [above]
        {{\begin{tabular}{c} $\rbf^{\pm}_{k1}-\zbf^0_1$ \\ $\pbf^{\pm}_{k1}$ \end{tabular}}}  (h2);
    \path[<->] (h2) edge  node [above]
        {{\begin{tabular}{c} $\rbf^{\pm}_{k2}-\zbf^0_2$ \\ $\qbf^{\pm}_{k2}$ \end{tabular}}} (V2b);
    \path[<->] (V2b) edge  node [above] {$\pbf^{\pm}_{k2}$} (h3);
    \path[<->] (h3) edge  node [above] {$\qbf^{\pm}_{k3}$} (V3b);
    \path[<->] (V3b) edge  node [above]
        {{\begin{tabular}{c} $\rbf^{\pm}_{k3}-\zbf^0_3$ \\ $\pbf^{\pm}_{k3}$ \end{tabular}}} (h4);
\end{tikzpicture}
\caption{Transformed view of the MLP and message passing system in Fig.~\ref{fig:nn_ml_vamp}.
The linear transforms $\Wbf_\ell$ are replaced by the SVD $\Wbf_\ell = \Vbf_\ell\Sigmabf_\ell\Vbf_{\lm1}$,
and intermediate layers are added for each component of the SVD.  With this transformation,
the MLP and message passing algorithm are reduced to alternating multiplications by $\Vbf_\ell$ and $\Vbf_\ell^*$
and componentwise (possibly nonlinear) functions.
 \label{fig:transformed}}
\end{figure*}

\subsection{Parameters}
To handle parameterized functions, the analysis in \cite{fletcher2017inference} introduces
the concept of \emph{parameter lists}.  For our purpose, let
\beq \label{eq:lamlist}
    \Lambda^+_{k\ell} := (\alpha_{k\ell}^+,\theta_{k\ell}^+), \quad
    \Lambda^-_{k\ell} := (\alpha_{k\ell}^-,\theta_{k\ell}^-),
\eeq
which is simply the parameter $\alpha_{k\ell}^{\pm}$ along with the parameter $\theta_{k\ell}^{\pm}$
for the estimators.

\subsection{Estimation Functions}
Similar to the transformed system in the top panel of Fig.~\ref{fig:transformed},
we next represent the steps in the
ML-VAMP Algorithm~\ref{algo:ml-vamp} as a sequence of alternating linear and nonlinear maps.
Let $\ell=0,2,4,\ldots,L$ be the index of  a nonlinear layer and
define the scalar-valued functions,
\begin{subequations} \label{eq:hnonlin}
\begin{align}
    \MoveEqLeft h^{+}_{0}(q_{0}^-,w_0,\theta_{k0}^+)
        := g^{+}_0(q^-_0 + w_0,\theta_{k0}^+), \label{eq:gtp0nonlin} \\
    \MoveEqLeft h^{+}_{\ell}(p^0_{\lm1},p_{\lm1}^+,q_{\ell}^-,w_\ell,\theta_{k\ell}^+) \nonumber \\
        &:= g^{+}_\ell(p_{\lm1}^+ +p^0_{\lm1},q^-_\ell + q^0_\ell,\theta_{k\ell}^+),  \label{eq:hpnonlin} \\
    \MoveEqLeft h^{-}_{L}(p^0_{\Lm1},p_{\Lm1}^+,w_L,\theta_{k,L}^-)
        := g^{-}_L(p_{\Lm1}^+ + p^0_{\Lm1},\theta_{k,L}^-). \label{eq:gtnLnonlin} \\
    \MoveEqLeft h^{-}_{\ell}(p^0_{\lm1},p_{\lm1}^+,q_{\ell}^-,w_\ell,\theta_{k,\lp1}^+) \nonumber \\
        &:= g^{-}_\ell(p_{\lm1}^+ +p^0_{\lm1},q^-_\ell + q^0_\ell,\theta_{k\ell}^-), \label{eq:hnnonlin}
\end{align}
\end{subequations}
For $\ell=1,3,\ldots,\Lm1$, the index of a linear layer, and
$w_\ell=(\bar{s}_\ell,\bar{b}_\ell)$, let
\begin{subequations} \label{eq:hlin}
\begin{align}
    \MoveEqLeft
    h^{+}_{\ell}(p_{\lm1}^0,p_{\lm1}^+,q_{\ell}^-,w_\ell,\theta_{k\ell}^+) \nonumber \\
        &:=  G^{+}_\ell(p_{\lm1}^+ +p^0_{\lm1},q^-_\ell + q^0_\ell, \bar{s}_\ell, \bar{b}_\ell,
        \theta^+_{k\ell}), \label{eq:hplin} \\
    \MoveEqLeft
    h^{-}_{\ell}(p_{\lm1}^0,p_{\lm1}^+,q_{\ell}^-,w_\ell,\theta_{k\ell}^-) \nonumber \\
        &:= G^{-}_\ell(p_{\lm1}^+ +p^0_{\lm1},q^-_\ell + q^0_\ell, \bar{s}_\ell, \bar{b}_\ell,
            \theta^-_{k\ell}), \label{eq:hnlin}
\end{align}
\end{subequations}
where $G^{\pm}(\cdot)$ are the components of the transformed linear estimation functions.
For both the linear and nonlinear layers, we then define the update functions as,
\begin{subequations} \label{eq:fhdef}
\begin{align}
    \MoveEqLeft f^{+}_{0}(q_{0}^-,w_0,\Lambda_{k0}^+) := \frac{1}{1-\alpha_{k\ell}^+} \nonumber \\
     &\times \left[
            h^{+}_{0}(q_{0}^-,w_0,\theta_{k0}^+) -w_0
            - \alpha_{k0}^+ q_0^- \right], \label{eq:fh0p} \\
    \MoveEqLeft f^{+}_{\ell}(p^0_{\lm1},p_{\lm1}^+,q_{\ell}^-,w_\ell,\Lambda_{k\ell}^+)
    := \frac{1}{1-\alpha_{k\ell}^+} \nonumber \\
    & \times \left[
            h^{+}_{\ell}(p^0_{\lm1},p_{\lm1}^+,q_{\ell}^-,w_\ell,\theta_{k\ell}^+)
            - q^0_\ell - \alpha_{k\ell}^+ Q_\ell^- \right], \label{eq:fhp} \\
    \MoveEqLeft f^{-}_{L}(p^0_{\Lm1},p_{\Lm1}^+,w_L,\Lambda_{kL}^-)
    := \frac{1}{1-\alpha_{k\ell}^-}\nonumber \\
    &\times \left[
            h^{-}_{L}(p^0_{\Lm1},p_{\Lm1}^+,w_L,\theta_{kL}^-) - p^0_{\Lm1}
            - \alpha_{k,\Lm1}^- p_{\Lm1}^+ \right], \label{eq:fhLn}    \\
    \MoveEqLeft f^{-}_{\ell}(p^0_{\lm1},p_{\lm1}^+,q_{\ell}^-,w_\ell,\Lambda_{k\ell}^-)
    := \frac{1}{1-\alpha_{k,\lm1}^-} \nonumber \\
     &\times \left[
            h^{-}_{\ell}(p^0_{\lm1},p_{\lm1}^+,q_{\ell}^-,w_\ell,\theta_{k\ell}^-)
            - p^0_{\lm1}
            - \alpha_{k,\lm1}^- p_{\lm1}^+ \right]. \label{eq:fhn}
\end{align}
\end{subequations}
With these definitions, let $\fbf^{\pm}_\ell(\cdot)$ and $\hbf^{\pm}_\ell(\cdot)$
be the componentwise extensions of
$f^{\pm}_\ell(\cdot)$ and $h^{\pm}_\ell(\cdot)$.
It is then shown in \cite{fletcher2017inference} that the vectors
in \eqref{eq:pqdef} satisfy the recursions in the ``Forward" and ``Reverse" passes
of the transformed ML recursion in Algorithm~\ref{algo:trans}.

This is diagrammatically represented in the bottom panel of Fig.~\ref{fig:transformed}.
We see that, in the forward pass, the vectors are generated by an alternating sequence
of componentwise mappings where
\[
    \qbf^+_{k\ell} = \fbf^+_\ell(\pbf^+_{k,\lm1},\qbf^-_{k\ell},\cdots),
\]
followed by multiplication by $\Vbf_\ell$,
\[
    \pbf^+_{k\ell} = \Vbf_\ell \qbf^+_{k\ell}.
\]
Similarly, in the reverse pass, we have a componentwise mapping,
\[
    \pbf^-_{k\ell} = \fbf^-_{\lp1}(\pbf^+_{k\ell},\qbf^-_{k,\lp1},\cdots),
\]
followed by multiplication by $\Vbf_\ell\tran$,
\[
    \qbf^-_{k\ell} = \Vbf_\ell\tran \pbf^-_{k\ell}.
\]
Thus, similar to the MLP, we have written the forward and reverse passes of the
multi-layer updates as alternating sequence of componentwise (possibly nonlinear) functions
followed by multiplications by orthogonal matrices.

\subsection{SE Analysis}

Now that the variables and the ML-VAMP algorithm estiamtes are written in
the form of Algorithm~\eqref{algo:trans}, the
analysis of \cite{fletcher2017inference} to derive a simple state evolution.
Let
\begin{subequations} \label{eq:wdefvar}
\begin{align}
    W_0 &:= Z^0_0, \quad
    W_{\ell} := \Xi_\ell, \quad \ell =2,4,\ldots,L \\
    W_{\ell} &= (\bar{S}_{\ell},\bar{B}_\ell,\bar{\Xi}_\ell),  \quad \ell =1,3,\ldots,\Lm1,
\end{align}
\end{subequations}
where  $Z^0_0$, $\Xi_\ell$ and $(\bar{S}_{\ell},\bar{B}_\ell,\bar{\Xi}_\ell)$
are the random variable limits in \eqref{eq:varinitnl} and \eqref{eq:varinitlin}.
With these definitions, we can recursively define the random variables $Q_{k\ell}^{\pm}$
and $P_{k\ell}^{\pm}$ from the steps in Algorithm~\ref{algo:gen_se}.
This recursive definition of random variables is called the \emph{state evolution}.
We see that the SE updates in Algorithm~\ref{algo:gen_se} are in a one-to-one
correspondence with the steps in Transformed ML-VAMP algorithm, Algorithm~\ref{algo:trans}.
The key difference is that the SE updates involves \emph{scalar} random variables, as opposed 
to vectors.  The random variables are all either Gaussian random variables or the output of 
nonlinear function of the Gaussian random variables.  
In addition, the parameters of the Gaussians such as $\Kbf^+_{k\ell}$ and $\tau^-_{k\ell}$
are fully deterministic since they are computed via expectations.  

We now make further assumption:

\begin{assumption} \label{as:alpha}
Let $\alphabar_{k\ell}^{\pm}$ be generated by the SE recursions in Algorithm~\ref{algo:gen_se}.
Then  $\alphabar^{\pm}_{k\ell} \in (0,1)$ from the for all $k$ and $\ell$.
\end{assumption}

We can now state the main result.  The result includes Theorem~\ref{thm:mapse} as a special case.

\begin{theorem} \label{thm:genConv}  Let $\wbf_\ell, \pbf^{\pm}_{k\ell}$,
$\qbf^{\pm}_{k\ell}$, $\pbf^0_\ell$, $\qbf^0_\ell$
be defined as above.  Consider the sequence of random variables defined
by the SE updates in Algorithm~\ref{algo:gen_se} under the above assumptions.
Then,
\begin{enumerate}[(a)]
\item For any fixed $k$ and $\ell=1,\ldots,\Lm1$,
the parameter list $\Lambda_{k\ell}^+$ converges as
\beq \label{eq:Lamplim}
    \lim_{N \arr \infty} \Lambda_{k\ell}^+ = \Lambdabar_{k\ell}^+
\eeq
almost surely.
Also, the components of
$\wbf_\ell$, $\pbf^0_{\lm1}$, $\qbf^0_{\ell}$, $\pbf_{0,\lm1}^+,\ldots,\pbf_{k,\lm1}^+$ and $\qbf_{0\ell}^\pm,\ldots,\qbf_{k\ell}^\pm$
almost surely empirically converge jointly with limits,
\begin{align}
   \MoveEqLeft \lim_{N \arr \infty} \left\{
        (p^0_{\lm1,n},p^+_{i,\lm1,n},q^0_{\ell,n},q^-_{j\ell,n},q^+_{j\ell,n}) \right\} \nonumber \\
        &=
        (P^0_{\lm1},P^+_{i,\lm1},Q^0_{\ell},Q^-_{j\ell}, Q^+_{j\ell}), \label{eq:PQplim}
\end{align}
for all $i,j=0,\ldots,k$, where the variables
$P^0_{\lm1}$, $P_{i,\lm1}^+$ and $Q_{j\ell}^-$
are zero-mean jointly Gaussian random variables independent of $W_\ell$ with
\begin{align} \label{eq:PQpcorr}
\begin{split}
    &\Cov(P^0_{\lm1},P_{i,\lm1}^+) = \Kbf_{i,\lm1}^+, \quad \Exp(Q_{j\ell}^-)^2 = \tau_{j\ell}^-, \nonumber \\
    &\Exp(P_{i,\lm1}^+Q_{j\ell}^-)  = 0,  \quad \Exp(P^0_{\lm1}Q_{j\ell}^-)  = 0,
\end{split}
\end{align}
The identical result holds for $\ell=0$ with the variables $\pbf_{i,\lm1}^+$ and $P_{i,\lm1}^+$ removed.

\item For any fixed $k > 0$ and $\ell=1,\ldots,\Lm1$,
the parameter lists $\Lambda_{k\ell}^-$ converge as
\beq \label{eq:Lamnlim}
    \lim_{N \arr \infty} \Lambda_{k\ell}^- = \Lambdabar_{k\ell}^-
\eeq
almost surely.
Also, the components of
$\wbf_\ell$, $\pbf^0_{\lm1}$, $\pbf_{0,\lm1}^+,\ldots,\pbf_{\km1,\lm1}^+$, $\pbf_{0,\lm1}^+,\ldots,\pbf_{\km1,\lm1}^+$, and $\qbf_{0\ell}^-,\ldots,\qbf_{k\ell}^-$
almost surely empirically converge jointly with limits,
\begin{align}
    \MoveEqLeft \lim_{N \arr \infty} \left\{
        (p^0_{\lm1,n},p^+_{i,\lm1,n},q^-_{j\ell,n},q^+_{j\ell,n}) \right\} \nonumber \\
        &= (P^0_{\lm1},P^+_{i,\lm1},Q^-_{j\ell}, Q^+_{j\ell}), \label{eq:PQnlim}
\end{align}
for all $i=0,\ldots,\km1$ and $j=0,\ldots,k$, where the variables
$P^0_{\lm1}$, $P_{i,\lm1}^+$ and $Q_{j\ell}^-$
are zero-mean jointly Gaussian random variables independent of $W_\ell$ with
\begin{align} \label{eq:PQncorr}
\begin{split}
    & \Cov(P^0_{\lm1},P_{i,\lm1}^+) = \Kbf_{i,\lm1}^+,\quad
    \quad \Exp(Q_{j\ell}^-)^2 = \tau_{j\ell}^-, \\
    & \Exp(P_{i,\lm1}^+Q_{j\ell}^-)  = 0,
    \quad \Exp(P^0_{\lm1}Q_{j\ell}^-)  = 0,
\end{split}
\end{align}
The identical result holds for $\ell=L$ with all the variables $\qbf_{j\ell}^-$ and $Q_{j\ell}^-$ removed.
Also, for $k=0$, we remove the variables with $\pbf_{\km1,\ell}^+$ and $P_{\km1,\ell}^+$.
\end{enumerate}
\end{theorem}
\begin{proof}
This is proven almost identically to the result in \cite{fletcher2017inference}.
\end{proof}

\section{Numerical Experiments Details} \label{sec:simdetails}

\paragraph*{Synthetic random network}  The simulation is identical to \cite{fletcher2017inference},
except that we have run MAP ML-VAMP intead of MMSE ML-VAMP.  The details of the simulation are as follows:
As described in Section~\ref{sec:sim},
the network input is a $N_0=20$ dimensional Gaussian unit noise vector $\zbf_0$.
and has three hidden layers with 100 and 500 units and a variable number $N_y$ of
output units.  For the weight matrices and bias vectors
in all but the final layer,
we took $\Wbf_\ell$ and $\bbf_\ell$ to be random i.i.d.\ Gaussians.  The mean of the bias vector was selected
so that only a fixed fraction, $\rho=0.4$, of the linear outputs would be positive.  The activation
functions were rectified linear units (ReLUs), $\phi_\ell(z)=\max\{0,x\}$.  Hence, after activation,
there would be only a fraction $\rho=0.4$ of the units would be non-zero.  In the final layer,
we constructed the matrix similar to \citep{RanSchFle:14-ISIT} where $\Abf=\Ubf\Diag(\sbf)\Vbf\tran$,
with $\Ubf$ and $\Vbf$ be random orthogonal matrices and $\sbf$ be logarithmically spaced valued
to obtain a desired condition number of $\kappa=10$.  It is known from \citep{RanSchFle:14-ISIT}
that matrices with high condition numbers are precisely the matrices in which AMP algorithms fail.
For the linear measurements,
$\ybf = \Abf\zbf_5+\wbf$, the noise level $10\log_{10}( \Exp\|\wbf\|^2/\|\Abf\zbf_5\|^2)$ is set at 30 dB.
In Fig.~\ref{fig:randmlp_sim}, we have plotted the normalized MSE (in dB) which we define as
\[
    \mathrm{NMSE} := 10\log_{10}\left[ \frac{\|\zbf^0_0-\hat{\zbf}^{\pm}_{k0}\|^2}{\|\zbf^0_0\|^2} \right].
\]
Since each iteration of ML-VAMP involves a forward and reverse pass, we say that each iteration
consists of two ``half-iterations", using the same terminology as turbo codes.  The left panel
of Fig.~\ref{fig:randmlp_sim} plots the NMSE vs. half iterations.

\paragraph*{MNIST inpainting}  The well-known MNIST dataset consists of handwritten images
of size $28 \x 28 = 784$ pixels. We followed the procedure in \citep{kingma2013auto} for training
a generative model from 50,000 digits.
Each image $\xbf$ is modeled as the output of a neural network
input dimension of 20 variables followed by
a single hidden layer with 400 units and an output layer of 784 units,
corresponding to the dimension of the digits.
ReLUs were used for activation functions and a sigmoid was placed at the output
to bound the final pixel values between 0 and 1.
The inputs $\zbf^0_0$ were the modeled as zero mean Gaussians with unit variance.
The data was trained using the Adam optimizer
with the default parameters in TensorFlow
\footnote{Code for the training was based on
\url{https://github.com/y0ast/VAE-TensorFlow} by Joost van Amersfoort.}
The training optimization was run
with 20,000 steps with a batch size of 100 corresponding to 40 epochs.

The ML-VAMP algorithm was compared against MAP estimation.
As studied in \cite{yeh2016semantic,bora2017compressed}, MAP 
estimation can be performed via numerical minimization of the likelihood.
In this study, 
We used TensorFlow for the minimization.  We found the fastest
convergence with the Adam optimizer at a step-size of 0.01.  This required
only 500 iterations to be within 1\% of the final loss function.
For MAP ML-VAMP, the sigmoid
function does not have an analytic denoiser, so it was approximated with a probit output.
We found that the basic MAP ML-VAMP algorithm could be unstable.  Hence, damping 
as described in \citep{RanSchFle:14-ISIT} and \cite{rangan2017vamp} was used.
With damping, we needed to run the ML-VAMP algorithm for up to 500 iterations, which 
is comparable to the Adam optimizer.
